\documentclass[11pt]{article}
\usepackage{fullpage}
\usepackage{times}
\usepackage[english]{babel}     
\usepackage[latin1]{inputenc}
\usepackage{graphics}
\usepackage{url}                            
\usepackage{amsmath}[1999/12/15] 
\usepackage{amssymb}
\usepackage[usenames,dvips]{color}

\newcommand{\pprotocol}[5]{{\begin{figure}[#4]
\begin{center}
\fbox{
\hbox{\quad
\begin{minipage}{0.9\textwidth}
#5
\end{minipage}
\quad} }
\caption{\label{#3} #2}
\end{center}
\end{figure} } }

\newcommand{\myprotocolh}[4]{\pprotocol{#1}{#2}{#3}{htp!}{#4}}

\newtheorem{definition}{Definition}[section]

\newtheorem{theorem}{Theorem}
\newtheorem{facten}[definition]{Fact}

\newtheorem{claim}[definition]{Claim}
\newtheorem{obs}[definition]{Observation}
\newtheorem{notat}[definition]{Notation}
\newtheorem{lemma}[definition]{Lemma}

\newtheorem{example}[definition]{Example}
\newtheorem{construction}[definition]{Construction}

\newcommand{\qed}{\hspace*{\fill} $\Box$}

\newenvironment{proof}{\noindent {\bf Proof:} \hspace{.677em}}{\qed}

\newcommand{\comp}{\stackrel{\mbox{\tiny C}}{{\equiv}}}

\newcommand{\indi}{\stackrel{\mbox{\tiny 1/p}}{{\approx}}}

\newcommand{\pr}{\Pr}

\newcommand{\poly}{\operatorname {poly}}
\newcommand{\xor}{\oplus}

\newcommand{\Real}{\operatorname{REAL}}

\newcommand{\Ideal}{\operatorname{IDEAL}}

\newcommand{\MultiShareGenDomain}{\operatorname{MultiShareGen}}
\newcommand{\MultiShareGenDomainWithAbort}{\operatorname{MultiShareGenWithAbort}}

\newcommand{\MPCWithDealer}{\operatorname{MPCWithD}}
\newcommand{\MPCWithDealerRange}{\operatorname{MPCWithDForRange}}
\newcommand{\FairMPC}{\operatorname{FairMPC}}

\newcommand{\MPC}{\operatorname{MPC}}
\newcommand{\MPCRange}{\operatorname{MPCForRange}}

\newcommand{\outValue}{{w}}
\newcommand{\outValueSim}{{w_S}}
\newcommand{\outValueVAR}{{C}}
\newcommand{\viewVAR}{{V}}

\newcommand{\outValueVar}{{c}}
\newcommand{\viewVar}{{v}}

\newcommand{\bitFiveP}[2]{{ \sigma_{#2} ^{#1} }}

\newcommand{\istar}{i^\star}
\newcommand{\partySet}{Q}

\newcommand{\badI}{B}

\newcommand{\abortedI}{D}
\newcommand{\abortedAfterPre}{D_0}

\newcommand{\Reconstruction}{\operatorname{Reconstruction}}
\newcommand{\abort}[1]{{``\operatorname*{\texttt{abort}}}_{#1}\textrm{''}}
\newcommand{\continueMSG}{``\operatorname{\texttt{continue}}\textrm{''}}
\newcommand{\proceedMSG}{``\operatorname{\texttt{proceed}}\textrm{''}}
\newcommand{\startMSG}{``\operatorname{\texttt{start}}\textrm{''}}

\def\partNum{{m}}
\def\badNum{{t}}
\def\goodNum{{m-t}}

\def\secParam{{n}}
\def\numRounds{{r}}

\def\3PartyBias{{8/r}}

\def\IOLvalue{{\sigma}}

\newcommand{\domainSize}{{d}}
\newcommand{\rangeSize}{{g}}
\newcommand{\simDealer}{{\S_T}}

\newcommand{\comment}[1]{}

\def\A{{\cal A}}

\def\D{{\cal D}}

\def\F{{\cal F}}

\def\J{{\cal J}}

\def\S{{\cal S}}

\newcommand{\NN}{{\mathbb N}}

\newcommand{\FF}{{\mathbb F}}

\newcommand{\defref}[1]{Definition~\ref{def:#1}}
\newcommand{\constructref}[1]{Construction~\ref{construct:#1}}
\newcommand{\lemref}[1]{Lemma~\ref{lem:#1}}

\newcommand{\secref}[1]{Section~\ref{sec:#1}}

\newcommand{\appref}[1]{Appendix~\ref{app:#1}}
\newcommand{\thmref}[1]{Theorem~\ref{thm:#1}}
\newcommand{\clmref}[1]{Claim~\ref{clm:#1}}

\newcommand{\eqnref}[1]{Equation~(\ref{eqn:#1})}

\newcommand{\figref}[1]{Figure~\ref{fig:#1}}

\newcommand{\stepref}[1]{Step~(\ref{stp:#1})}
\newcommand{\steprefs}[2]{Steps~\ref{stp:#1}--\ref{stp:#2}}
\newcommand{\set}[1]{\left\{#1\right\}}

\newcommand{\remove}[1]{}

\newcommand{\vect}[1]{(#1)}
\newcommand{\size}[1]{\left| #1 \right|}

\newcommand{\valueInnerSecretSharingShare}[3]{{ S^{#1,#2}_{{#3}}}}
\newcommand{\valueInnerSecretSharingShareSigned}[3]{{  R_{{#3}}^{#1,#2}}}

\newcommand{\valueInnerSecretSharingShareFromT}[3]{{ X^{#1,#2}_{{#3}}}}
\newcommand{\valueInnerSecretSharingShareSignedFromT}[3]{{  Y_{{#3}}^{#1,#2}}}

\newcommand{\superSet}[1]{{Q_{#1}}}

\newcommand{\masking}[2]{\operatorname{mask}_{#2}(#1)}
\newcommand{\compl}[1]{\operatorname{comp}(#1)}
\newcommand{\complParty}[2]{\operatorname{comp}_{#2}(#1)}
\newcommand{\win}[1]{\operatorname{win}(#1)}
\newcommand{\SD}[2]{\operatorname*{SD}\left(#1,#2\right)}

\newcommand{\valueMultiParty}[2]{{ \sigma_{#2}^{#1} }}
\newcommand{\valueMultiPartyFromT}[2]{{ \tau_{#2}^{#1} }}

\newcommand{\randomDomain}[2]{{\widehat{x}_{#1}}}

\newcommand{\partyInput}[1]{{x_{#1}}}
\newcommand{\initialInput}[1]{{y_{#1}}}
\newcommand{\partyDomain}{{X_{\secParam}}}
\newcommand{\partyRange}{Z_{\secParam}}
\newcommand{\domainLength}{{\ell_d}}
\newcommand{\rangeLength}{{\ell_r}}
\newcommand{\vecInput}{{\vec{y}}}
\newcommand{\range}[1]{[#1]}
\newcommand{\aux}{{\rm aux}}
\begin{document}

\begin{titlepage}
\title{Secure Multiparty Computation with Partial Fairness}
\author{ Amos Beimel\thanks{Supported by ISF grant 938/09 and by the Frankel Center for Computer Science.} \\
   Department of Computer Science \\ Ben Gurion University \\
   Be'er Sheva, Israel
   \and Eran Omri\thanks{This research was generously supported by the European Research Council as part of the ERC project ``LAST''.} \\
   Department of Computer Science \\ Bar Ilan University \\
   Ramat Gan, Israel
   \and Ilan Orlov\thanks{Supported by ISF grant 938/09 and by the Frankel Center for Computer Science.}  \\
   Department of Computer Science \\ Ben Gurion University \\
   Be'er Sheva, Israel
 }
    
\maketitle

\begin{abstract}
A protocol for computing a functionality is secure if an adversary in this
protocol cannot cause more harm  than in an ideal computation where parties
give their inputs to a trusted party which returns the output of the
functionality to all parties. In particular, in the ideal model such
computation is fair -- all parties get the output.  Cleve (STOC 1986)
proved that, in general, fairness is not possible without an honest
majority. To overcome this impossibility, Gordon and Katz (Eurocrypt 2010)
suggested a relaxed definition -- $1/p$-secure computation -- which
guarantees partial fairness. For two parties, they construct $1/p$-secure
protocols for functionalities for which the size of either their domain or their range is polynomial (in the security parameter).  Gordon and Katz ask whether their results can be extended to multiparty protocols.

We study $1/p$-secure protocols in the multiparty setting
for general functionalities.
Our main result is constructions of $1/p$-secure protocols when  the number of parties is
constant provided that less than 2/3 of the parties are corrupt.  Our
protocols require that either 
(1) the functionality is deterministic and
the size of the domain is polynomial (in the security parameter), or 
(2) the functionality can be randomized and the size of the range is polynomial.  
If the size of the domain is constant and the functionality is deterministic, then our
protocol is efficient even  when the number of parties is $O(\log \log
\secParam)$ (where $\secParam$ is the security parameter).  On the negative
side, we show that when the number of parties is super-constant,
$1/p$-secure protocols are not possible when the size of the domain is
polynomial.
\end{abstract}

\thispagestyle{empty}
\end{titlepage}


\section{Introduction}
\label{sec:intro}

A protocol for computing a functionality is secure if an adversary in this protocol cannot cause more harm than in an ideal computation where
parties give their inputs to a trusted party which returns the output of
the functionality to all parties. This is formalized by requiring that for every adversary in the real world, there is an adversary in the ideal
world, called simulator, such that the output of the real-world adversary
and the simulator are indistinguishable in polynomial time. Such security
can be achieved when there is a majority of honest
parties~\cite{GMW87}. Secure computation is fair -- all parties get the
output. Cleve~\cite{Cle86} proved that, in general, fairness is not
possible without an honest majority.

To overcome the impossibility of~\cite{Cle86}, Gordon and Katz~\cite{GK10} suggested a relaxed definition -- $1/p$-secure computation -- which
guarantees partial fairness. Informally, a protocol is $1/p$-secure if for
every adversary in the real world, there is a simulator running in the
ideal world, such that the output of the real-world adversary and the
simulator cannot be distinguished with probability greater than $1/p$. For two parties, Gordon and Katz construct $1/p$-secure protocols for functionalities whose size of either their domain or their range is polynomial (in the
security parameter). They also give impossibility results when both the
domain and range are super-polynomial. Gordon and Katz ask whether their
results can be extended to multiparty protocols. We give positive and
negative answers to this question.

\paragraph{Previous Results.}
Cleve~\cite{Cle86} proved that any protocol for coin-tossing without an
honest majority cannot be fully secure, specifically, if the protocol has
$\numRounds$ rounds, then it is at most $1/\numRounds$-secure. Protocols
with partial fairness, under various definitions and assumptions, have been
constructed for coin-tossing~\cite{Cle86,Cle90,MNS09,BOO10}, for contract
signing/exchanging secrets~\cite{Bl84,LMR83,EGL85,BGMR90,Dam95,BN00}, and
for general functionalities~\cite{Yao86,GHY87,BG89,GL90,Pin03,GMPY06,GK10}.
We next describe the papers that are most relevant to our paper.  Moran,
Naor, and Segev~\cite{MNS09} construct 2-party protocols for coin tossing
that are $1/\numRounds$-secure (where $\numRounds$ is the number of rounds
in the protocol). Gordon and Katz~\cite{GK10} define $1/p$-security and
construct 2-party $1/p$-secure protocols for every functionality whose
size of either the domain or the range of the functionality is polynomial.
Finlay, in a previous work~\cite{BOO10} we construct multiparty protocols
for coin tossing that are $O(1/\numRounds)$-secure provided that the
fraction of bad parties is slightly larger than half. In particular, our
protocol is $O(1/\numRounds)$-secure when the number of parties is constant
and the fraction of bad parties is less than 2/3.

Gordon et al.~\cite{GHKL08} showed that complete fairness is possible in the
two party case for some functions. Gordon and Katz~\cite{GK09} showed
similar results for the multiparty case. The characterization of the
functions that can be computed with full fairness without honest majority is open.
Completeness for fair computations has been studied in~\cite{GIMOS10}. 
Specifically, they show a specific function that is complete for fair two-party computation; this function is also complete for $1/p$-secure two-party computation.

\subsection{Our Results}
We study $1/p$-secure protocols in the multiparty setting.  We construct
two protocols for general functionalities assuming that the fraction of
corrupt parties is less than 2/3.  The first protocol is efficient when
(1) The number of parties is constant, the functionality is deterministic,
and the size of the domain of inputs is at most polynomial in the security
parameter, or (2) The number of parties is $O(\log\log\secParam)$ (where
$\secParam$ is the security parameter), the functionality is deterministic,
and the size of the domain of inputs is constant.
The second protocol is efficient when
the number of parties is constant, the functionality can be
randomized, and the size of the {\em range} of the functionality
is at most polynomial in the security parameter.
Our second protocol does not provide correctness, i.e., in a case of premature termination,
with probability of $1/\poly(\secParam)$, the remaining active parties output a value which might be inconsistent with their inputs.
In contrast, our first protocol provides correctness.

Our protocols combine ideas from the protocols of Gordon and
Katz~\cite{GK10} and our paper~\cite{BOO10}, both of which generalize the
protocol of Moran, Naor, and Segev~\cite{MNS09}. Specifically, our
protocols proceed in rounds, where in each round values are given to subsets
of parties. There is a special round $\istar$ in the protocol. Prior to
round $\istar$, the values given to a subset of parties are values that can  be computed from the inputs of the parties in this subset; staring from round $\istar$ the values are the
``correct'' output of the functionality. The values given to a subset are secret
shared such that only if all parties in the subset cooperate they can
reconstruct the value. If in some round many (corrupt) parties have
aborted such that there is a majority of honest parties among the active
parties, then the set of active parties reconstructs the value given to
this set in the previous round.\footnote{As parties can abort during this reconstruction, they actually
  reconstruct the value of a subset of this set.}
Similar to the protocols of~\cite{MNS09,GK10,BOO10}, the adversary can cause harm (e.g., bias
the output of the functionality) only if it guesses $\istar$; we show that in
our protocols this probability is small and the protocols are $1/p$-secure.
The values in our protocols are chosen similar to~\cite{GK10}. The mechanism to
secret share the values is similar to~\cite{BOO10}, however, there are
important differences in this sharing, as the sharing mechanism
of~\cite{BOO10} is not appropriate for $1/p$-secure computations of
functionalities which depend on inputs.

To complete the picture, we prove interesting impossibility results.  We
show that, in general, when the number of parties is super-constant,
$1/p$-secure protocols are not possible without honest majority when the
size of the domain is polynomial. This impossibility result justifies the
fact why in our protocols the number of parties is constant.  We also show
that, in general, when the number of parties is $\omega(\log \secParam)$,
$1/p$-secure protocols are not possible without honest majority even when
the size of the domain is 2. The proof of the impossibility result is rather
simple and follows from an impossibility result of~\cite{GK10}.

Our impossibility results should be contrasted with the coin-tossing
protocol of~\cite{BOO10} which is an efficient $1/p$-secure protocol even
when $\partNum(\secParam)$, the number of parties, is polynomial in the security parameter and the number of
bad parties is $\partNum(\secParam)/2+O(1)$. Our results show that these parameters are not
possible for general $1/p$-secure protocols even when the size of the
domain of inputs is 2.

\paragraph{Open Problems.}
In both our impossibility results the size of the range is
super-polynomial. It is open if there is an efficient $1/p$-secure protocol when the number of parties is not constant and the size of both the domain and range is polynomial. In addition, the impossibility results do not rule
out that the double-exponential dependency on the number of parties can be improved.

The protocols of~\cite{GK10} are private -- the adversary cannot learn any
information on the inputs of the honest parties (other than the information
that it can learn in the ideal world of computing $\F$). The adversary can
only bias the output. Our first protocol is not private (that is, the adversary can learn extra information).
However, we do not know whether the second protocol is private.\footnote{The problem in our protocols is that the adversary can keep one
corrupted party active, thus, the adversary can get the output of the
honest parties.}
It is open if there are general multiparty $1/p$-secure protocols that
are also private.

\section{Preliminaries}
\label{sec:prelim}
A multi-party protocol with $\partNum$ parties is defined by $\partNum$
interactive probabilistic polynomial-time Turing machines
$p_1,\ldots,p_\partNum$. Each Turning machine, called party, has the security parameter $1^\secParam$ as a
joint input and a private input $\initialInput{j}$.  The
computation proceeds in rounds. In each round, the active parties broadcast and receive messages on a common broadcast channel.  The number of rounds in the protocol is expressed as some function $\numRounds(\secParam)$ in the
security parameter (typically, $\numRounds(\secParam)$ is bounded by a polynomial).
At the end of the protocol, the (honest) parties should hold a common value
$\outValue$ (which should be equal to an output of a predefined
functionality).

In this work we consider a corrupt, static, computationally-bounded (i.e., non-uniform probabilistic polynomial-time) adversary that is allowed to corrupt some subset of parties. That is, before the beginning of the protocol, the adversary corrupts a subset of the parties and may instruct them to deviate from the protocol in an arbitrary way. The adversary has complete access to the internal state of each of the corrupted parties and fully controls the messages that they send throughout the protocol. The honest parties follow the instructions of the protocol.

The parties communicate via a synchronous network, using only a broadcast
channel. The adversary is rushing, that is, in each round the adversary
hears the messages broadcast by the honest parties before broadcasting the
messages of the corrupted parties for this round (thus, broadcast messages
of the corrupted parties can depend on the broadcast messages of the honest
parties in this round).  

\paragraph{Notation.}
For an integer $\ell$, define $[\ell]=\set{1,\dots,\ell}$.  For a set $J
\subseteq \range{\partNum}$, define $\superSet{J} = \set{p_j : j\in J}$.
An $\partNum$-party functionality $\F=\set{f_\secParam}_{\secParam \in
\NN}$ is a sequence of polynomial-time computable, randomized mappings
$f_\secParam:(\partyDomain)^\partNum \rightarrow \partyRange$,
where $\partyDomain=\set{0,1}^{\domainLength(\secParam)}$ and
$\partyRange=\set{0,1}^{\rangeLength(\secParam)}$ are the domain of inputs of
each party and the range respectively;
$\domainLength,\rangeLength:\NN\rightarrow\NN$ are some fixed functions.
We denote the size of the domain and the range of $\F$ by
$\domainSize(\secParam)$ and $\rangeSize(\secParam)$ respectively, that is,
$\domainSize(\secParam)=2^{\domainLength(\secParam)}$ and
$\rangeSize(\secParam)=2^{\rangeLength(\secParam)}$. For a
randomized mapping $f_\secParam$, the assignment $\outValue\gets
f_\secParam(x_1,\dots,x_\partNum)$ denotes the process of computing
$f_\secParam$ with the inputs $x_1,\dots,x_\partNum$ and with uniformly
chosen random coins and assigning the output of the computation to
$\outValue$. If $\F$ is deterministic, we sometimes call it a function.
We sometime omit $\secParam$ from functions of $\secParam$ 
(for example, we write $\domainSize$ instead of $\domainSize(\secParam)$).

\subsection{The Real vs.~Ideal Paradigm}
\label{sec:realIdeal}
The security of multiparty computation protocols is defined using the real
vs.~ideal paradigm. In this paradigm, we consider the real-world model, in
which  protocols are executed. We then formulate an ideal model for
executing the task. This ideal model involves a trusted party whose
functionality captures the security requirements from the task. Finally, we
show that the real-world protocol ``emulates''  the ideal-world  protocol:
For any real-life adversary $\A$ there  exists an ideal-model
adversary $\S$ (called simulator) such that the global output of an
execution of the protocol    with $\A$ in the real-world model is
distributed similarly to the global output of running  $\S$ in the ideal
model.
In both models there are $\partNum$ parties $p_1,\ldots,p_\partNum$ holding  
a common input $1^\secParam$ and private inputs $\initialInput{1},\ldots,\initialInput{\partNum}$
respectively, where $\initialInput{j} \in \partyDomain{}$ for $1\leq j \leq \partNum$.
\paragraph{The Real Model.}
Let $\Pi$ be an $\partNum$-party protocol computing $\F$. Let $\A$ be a
non-uniform  probabilistic polynomial time adversary that gets the input
$\initialInput{j}$ of each corrupted party $p_j$ and the auxiliary input
$\aux$.  Let $\Real_{\Pi,\A(\aux)}(\vecInput,1^\secParam)$, where $\vecInput=(\initialInput{1},\ldots,\initialInput{\partNum})$, be the
random variable consisting of the view of the adversary (i.e., the inputs of the corrupted parties and the messages it got) and the output of  the honest parties
following an execution of $\Pi$.

\paragraph{The Ideal Model.}
The basic ideal model we consider is a model without abort. 
Specifically, there is an adversary $\S$ which has corrupted a subset $B$ of the parties. 
The adversary $\S$ has some auxiliary input $\aux$.
An ideal execution for the computing $\F$ proceeds as follows:

\begin{description}

\item[Send inputs to trusted party:] The honest parties send their inputs to the  trusted party. 
The corrupted parties may either send their received input, or send some other input of the same
length (i.e., $\partyInput{j} \in \partyDomain{}$) to the trusted party,
or abort (by sending a special $\abort{j}$ message).             
Denote by $\partyInput{1},\ldots,\partyInput{\partNum}$ the
inputs received by the trusted party.
If $p_j$ does not send an input, then the trusted party
selects $\partyInput{j} \in \partyDomain{}$ with uniform distribution.\footnote{For the simplicity of the presentation of our protocols, we present a slightly different ideal world than the traditional one.
In our model there is no a default input in case of an ``abort''. However, the protocol can be presented in the traditional model, where a predefined default input is used if a party aborts.}

\item[Trusted party sends outputs:] The trusted party computes $f_\secParam(\partyInput{1},\ldots,\partyInput{\partNum})$
with uniformly random coins and sends the output to the parties.
\item[Outputs:] The honest parties output the value sent by the trusted party, 
the corrupted parties output nothing, and $\S$ outputs any arbitrary 
(probabilistic polynomial-time computable) function of its view (its inputs, the output, and the auxiliary input $\aux$).
\end{description}

Let $\Ideal_{\F,\S(\aux)}(\vecInput,1^\secParam)$ be the 
random variable consisting of the output of the adversary $\S$ in this ideal world
execution and the output of the honest parties in the execution.

\subsubsection{\boldmath{$1/p$}-Indistinguishability and \boldmath{$1/p$}-Secure Computation}
As explained in the introduction, some ideal functionalities for computing $\F$ cannot be implemented when there is no honest majority. 
We use $1/p$-secure computation, defined by~\cite{GK10}, to 
capture the divergence from the ideal worlds. 
\begin{definition}[\boldmath{$1/p$}-indistinguishability] 
A function $\mu(\cdot)$ is \emph{negligible} 
if for every positive polynomial $q(\cdot)$ 
and all sufficiently large $n$ it holds that $\mu(n) < 1/q(n)$.
A \emph{distribution ensemble} $X = \left\{X_{a,\secParam}\right\}_{a\in \D_\secParam, \secParam \in \NN}$ 
is an infinite sequence of random variables indexed by $a\in \D_\secParam$ and $n \in \NN$,
where $\D_\secParam$ is a domain that might depend on $\secParam$.
For a fixed function $p(\secParam)$, 
two distribution ensembles $X = \{X_{a,\secParam}\}_{a \in \D_\secParam,n\in \NN}$ and 
$Y = \{Y_{a,\secParam}\}_{a \in \D_\secParam,n\in \NN}$
are \emph{computationally $1/p$-indistinguishable}, denoted $X \indi Y$,
if for every non-uniform polynomial-time algorithm $D$ there exists a negligible function $\mu(\cdot)$ such that for every $n$ and every $a \in \D_n$,
$$ \Big| \pr[D(X_{a,\secParam}) = 1] - \pr[D(Y_{a,\secParam}) = 1] \Big| \leq \frac{1}{p(n)} + \mu(n).$$
\end{definition}
Two distribution ensembles are \emph{computationally indistinguishable},
denoted $X \comp Y$, if for every $c \in \NN$ they are computationally $\frac{1}{n^c}$ -indistinguishable.
\medskip

We next define the
notion of $1/p$-secure computation~\cite{GK10}.
The definition uses the standard real/ideal paradigm~\cite{Gol04,Can00},
except that we consider a completely fair ideal model 
(as typically considered in the setting of honest majority),
and require only $1/p$-indistinguishability rather than indistinguishability. 

\begin{definition}[\boldmath{$1/p$}-secure computation~\cite{GK10}]
\label{def:1Overp-security}
Let $p = p(n)$ be a function.  An $\partNum$-party protocol $\Pi$ is said
to  $1/p$-securely compute a functionality $\F$
where there are at most $\badNum(\secParam)$ corrupt parties, if for every non-uniform
probabilistic polynomial-time adversary $\A$ in the real model controlling 
at most $\badNum(\secParam)$ parties, there
exists a non-uniform probabilistic polynomial-time adversary $\S$  in the
ideal model, controlling the same parties as $\A$, such that the following
two distribution ensembles are computationally $1/p$-indistinguishable
$$
\set{\Ideal_{\F,\S(\aux)}(\vecInput,1^\secParam)}_{\aux \in \set{0,1}^*,\vecInput \in (\partyDomain)^\partNum,\secParam\in\NN}
\quad\indi\quad
  \set{\Real_{\Pi,\A(\aux)}(\vecInput,1^\secParam)}_{\aux \in \set{0,1}^*,\vecInput \in (\partyDomain)^\partNum,\secParam\in\NN}.
  $$
\end{definition}

We next define statistical distance between two random variables and the notion of perfect $1/p$-secure computation, which implies the notion of $1/p$-secure computation.

\begin{definition}[statistical distance] 
\label{def:statisticalDistance}
We define the \emph{statistical distance} between two random variables $A$ and $B$
as the function
$$\SD{A}{B} = \frac{1}{2}\sum_{\alpha}{\Big| \pr[A = \alpha] - \pr[B = \alpha] \Big|}.$$
\end{definition}
\begin{definition}[perfect \boldmath{$1/p$}-secure computation]
\label{def:1Overp-perfect-security}
An $\partNum$-party protocol $\Pi$ is said to perfectly $1/p$-secure compute a functionality $\F$ if 
for every non-uniform adversary $\A$ in the real model, there exists a polynomial-time adversary $\S$ 
in the ideal model such that for every $\secParam\in\NN$, for every $\vecInput \in (\partyDomain)^\partNum$, and for every ${\rm aux} \in\set{0,1}^*$
$$\SD{\Ideal_{\F,\S({\rm aux})}(\vecInput,1^\secParam)}{\Real_{\Pi,\A({\rm aux})}(\vecInput,1^\secParam)} \leq \frac{1}{p(n)}.$$
\end{definition}

Security with abort and cheat detection is defined in \appref{securityWithBlaming}.
The cryptographic tools we use are described in \appref{cryptoTools}.

\section{The Multiparty Secure Protocols}
\label{sec:MPC}
In this section we present our protocols. We start with a protocol that
assumes that either the functionality is deterministic and the size of the domain is polynomial, or that the functionality is randomized and both the domain and range of the functionality are polynomial. We then present a modification of the protocol that is $1/p$-secure 
for (possibly randomized) functionalities
if the size of the range is polynomial (even if the size of the domain of $\F$ is not polynomial).
The first protocol is more efficient for deterministic functionalities with polynomial-size domain.  
Furthermore, the first protocol has full correctness, while in the modified protocol, correctness is only guaranteed with probability $1-1/p$. 

Formally, we prove the following two theorems.
\begin{theorem}
\label{thm:mainDomain}
Let $\F = \set{f_\secParam:(\partyDomain)^\partNum \rightarrow \partyRange}$ be randomized 
functionality where the size of domain is $\domainSize(\secParam)$ and the size of the range is $\rangeSize(\secParam)$,
and let $p(\secParam)$ be a polynomial.
If enhanced trap-door permutations exist, then for any $\partNum$ and $\badNum$
such that $\partNum/2 \leq \badNum < 2\partNum/3$, 
and for any polynomial $p(\secParam)$ there is an
$\numRounds(\secParam)$-round $\partNum$-party $1/p(\secParam)$-secure protocol computing $\F$
tolerating up to $\badNum$ corrupt parties where 
$\numRounds(\secParam) = p(\secParam) \cdot \left(2 \cdot \domainSize(\secParam)^\partNum \cdot \rangeSize(\secParam) \cdot p(\secParam)\right)^{ 2^\badNum}$,
provided that $\numRounds(\secParam)$ is bounded by a polynomial in $\secParam$.
If $\F$ is deterministic, then there is a $\numRounds(\secParam)$-round  $1/p(\secParam)$-secure protocol for 
$\numRounds(\secParam) = p(\secParam) \cdot \domainSize(\secParam)^{\partNum \cdot 2^\badNum}$,
provided that $\numRounds(\secParam)$ is bounded by a polynomial in $\secParam$.
\end{theorem}

\begin{theorem}
\label{thm:mainRange}
Let $\F = \set{f_\secParam:(\partyDomain)^\partNum \rightarrow \partyRange}$ be randomized 
functionality where the size of the range $\rangeSize(\secParam)$ is polynomial in $\secParam$ and $\partNum$ is constant,
and let $p(\secParam)$ be a polynomial.
If enhanced trap-door permutations exist, then for $\badNum$ such that $\partNum/2 \leq \badNum < 2\partNum/3$ 
and for any polynomial $p(\secParam)$ there is an
$\numRounds(\secParam)$-round $\partNum$-party $1/p(\secParam)$-secure protocol computing $\F$
tolerating up to $\badNum$ corrupt parties where 
 $\numRounds(\secParam) = \Big((2p(\secParam))^{  2^\badNum +1} \cdot \rangeSize(\secParam)^{  2^\badNum} \Big)$.
\end{theorem}

Following~\cite{MNS09,BOO10}, we present the first protocol in two stages. We first describe in \secref{MPCwithDealer} a protocol with a dealer and then in \secref{eliminatingDealer} present a protocol without this dealer. The goal of presenting the protocol in  two stages is to simplify the understanding of the protocol and to enable to prove the protocol in a modular way.  In \secref{polynomialRange}, we present a modification of the protocol which is $1/p$-secure if the size of the range is polynomial
(even if the size of the domain of $f$ is not polynomial).

\subsection{The Protocol for Polynomial-Size Domain with a Dealer}
\label{sec:MPCwithDealer}

We consider a network with $\partNum$ parties where at most $\badNum$
of them are corrupt such that $\partNum/2 \leq \badNum \leq 2\partNum/3$.
In this section we assume that  there is a special trusted on-line dealer, denoted $T$. This dealer interacts with the parties
in rounds, sending messages on private channels.  We assume that the dealer
knows the set of corrupt parties. In \secref{eliminatingDealer}, we show how to remove this
dealer and construct a protocol without a dealer.

In our protocol the dealer sends in each round values to subsets of parties;
the protocol proceeds with the normal
execution as long as at least $\badNum+1$ of the  parties are still
active. If at some round $i$, there are at most $\badNum$ active parties,
then the active parties reconstruct the value given to them in round
$i-1$, output this value, and halt.  Following~\cite{MNS09}, and its
follow up works~\cite{GK10,BOO10}, the
dealer chooses at random with uniform distribution a special round $\istar$. Prior to this round the
adversary gets no information and if the corrupt parties abort the
execution prior to $\istar$, then they cannot bias the output of the
honest parties or cause any harm. After round $\istar$, the output of the protocol is fixed,
and, also in this case the adversary cannot affect the output of the
honest parties. The adversary cause harm 
only if it guesses $\istar$ and this happens with small probability.

\myprotocolh{}
{Protocol $\MPCWithDealer_\numRounds$.}{fig:MPCWithDealer}{
\begin{description} 
\item[]{\hspace*{-0.63cm} \bf Inputs:} 
        Each party $p_j$ holds a private input $\initialInput{j} \in \partyDomain{}$ and the joint input:     
  the security parameter $1^\secParam$, the number of rounds $\numRounds=\numRounds(\secParam)$,
                and a bound $\badNum$ on the number of corrupted parties.

\item[] {\hspace*{-0.63cm} \bf Instructions for each honest party $p_j$:} 
 (1) After receiving the $\startMSG$ message, send input $\initialInput{j}$ to the dealer.
 (2) If the premature termination step is executed with $i=1$, then send its input $\initialInput{j}$ to the dealer.
 (3) Upon receiving output $z$ from the dealer, output $z$. 
        (Honest parties do not send any other messages throughout the protocol.)

\item[] {\hspace*{-0.63cm} \bf Instructions for the (trusted) dealer:}

\item[The preprocessing phase:] \quad
\begin{enumerate}
\item Set $\abortedAfterPre = \emptyset$ and send a $\startMSG$ message to all parties.
\item
Receive an input, denoted $\partyInput{j}$, from each party $p_j$.
For every $p_j$ that sends an $\abort{j}$
message, notify all parties that party $p_j$ aborted, select $\partyInput{j} \in \partyDomain{}$ with uniform distribution, and update $\abortedAfterPre = \abortedAfterPre \cup \set{j}$.

\item Let $\abortedI = \abortedAfterPre$. If $\size{\abortedI} \geq \goodNum$, go to premature termination with $i=1$.
\item \label{stp:calcOutput}Set $\outValue \gets f_\secParam(\partyInput{1},\ldots,\partyInput{\partNum})$ and select $\istar \in \set{1, \ldots, \numRounds}$ with uniform distribution.
\item \label{stp:beforeIStar1}For each $1\le i < \istar$, for each $J \subseteq
\range{\partNum}\setminus \abortedAfterPre$ s.t. $\goodNum \leq \size{J} \leq \badNum$:
        for each $j\in J$ set $\randomDomain{j}{i} = \partyInput{j}$,
        for each $j\not\in J$ select uniformly at random $\randomDomain{j}{i} \in \partyDomain{}$, and set $\valueMultiParty{i}{J} \gets
f_\secParam(\randomDomain{1}{i},\ldots,\randomDomain{\partNum}{i})$.
 
\item For each $\istar\le i \le \numRounds$ and for each
$J \subseteq \range{\partNum}\setminus \abortedAfterPre$ s.t. $\goodNum \leq \size{J} \leq \badNum$,
set $\valueMultiParty{i}{J}= \outValue$.
\item Send $\proceedMSG$ to all parties.
\end{enumerate}

\item[Interaction rounds:] In each round $1\le i\le \numRounds$, interact with 
  the parties in three phases:
  \begin{itemize}
        \item {\bf The peeking phase:} 
                For each $J \subseteq \range{\partNum}\setminus \abortedAfterPre$ s.t. $\goodNum \leq \size{J} \leq \badNum$,
                if $\superSet{J}$ contains only corrupt parties,
                send the value $\valueMultiParty{i}{J}$ to all parties in $\superSet{J}$.
        \item   {\bf The abort phase:}  Upon receiving an $\abort{j}$
                message from a party $p_j$, notify all parties  that party
                $p_j$ aborted (ignore all other types of messages) and update
                $\abortedI  = \abortedI  \cup \set{j}$.
                If $\size{\abortedI } \geq \goodNum$, go to  premature termination step.
        \item {\bf The main phase:} Send $\proceedMSG$ to all parties.
  \end{itemize}

\item[Premature termination step:]\quad
\begin{itemize}
                \item If $i=1$, then: Receive an input, denoted $\partyInput{j}'$, from each active party $p_j$.
                                        For every party $p_j$ that sends an $\abort{j}$ message,
                                        update $\abortedI  = \abortedI  \cup \set{j}$ and
                                        select $\partyInput{j}' \in \partyDomain{}$ with uniform distribution. Set $\outValue' \gets f_\secParam(\partyInput{1}',\ldots,\partyInput{\partNum}')$.
                \item Else, if $i>1$, then: For each $\abort{j}$ message received from a party $p_j$, update $\abortedI  = \abortedI  \cup \set{j}$.
        Set $\outValue' = \valueMultiParty{i-1}{J}$ for $J = \range{\partNum} \setminus \abortedI $.
                        
                \item Send $\outValue'$  to each party $p_j$ s.t. $j\notin \abortedAfterPre$ and halt.
\end{itemize}

\noindent  
\item[Normal termination:]
If the last round of the protocol is completed, send $\outValue$ to to each party $p_j$ s.t. $j\notin \abortedAfterPre$ . 

\end{description}
}

We next give a verbal description of the protocol.
This protocol is designed such that the dealer can be removed from it in \secref{eliminatingDealer}. 
A formal description is given in \figref{MPCWithDealer}.
At the beginning of the protocol each party sends its input $\initialInput{j}$ to the dealer. The corrupted parties may send any values of their choice. Let $\partyInput{1},\ldots,\partyInput{\partNum}$ denote the inputs received by the dealer.
If a corrupt party $p_j$ does not send its input, then the dealer sets $\partyInput{j}$
to be a random value selected uniformly from $\partyDomain{}$.
In a preprocessing phase, the dealer $T$ 
selects uniformly at random a special round $\istar \in [\numRounds]$.
The dealer computes $\outValue\gets f_\secParam(\partyInput{1},\dots,\partyInput{\partNum})$.
Then, for every round $1\leq i < \numRounds$  and every $J \subseteq \set{1, \ldots, \partNum}$ such that
$\partNum-\badNum \leq |J| \leq \badNum$, the dealer selects an output, denoted 
$\bitFiveP{i}{J}$, as follows (this output is returned by the parties in $\superSet{J} = \set{p_j : j\in J}$ if 
the protocol terminates in round $i+1$ and $\superSet{J}$ is the set of the active
parties):
\begin{description}
\item[\sc Case I: $1\leq i < \istar$.]
  For every $j \in J$ the dealer sets $\randomDomain{j}{i}=\partyInput{j}$ and for every $j \notin J$ it chooses
$\randomDomain{j}{i}$ independently with uniform distribution from the domain $\partyDomain$;
it computes the output $\valueMultiParty{i}{J} \gets f_\secParam(\randomDomain{1}{i},\dots,\randomDomain{\partNum}{i})$.
\item[\sc Case II: $\istar \leq i \leq \numRounds$.]
The dealer sets $\bitFiveP{i}{J}=\outValue$.
\end{description}

The dealer $T$ interacts with the parties in rounds, where in round $i$,
for $1 \leq i \leq \numRounds$, there are of three phases:
\begin{description}
\item[The peeking phase.]
The dealer sends to the adversary
all the values $\bitFiveP{i}{J}$ such that all parties in $\superSet{J}$ are corrupted. 
\item[The abort and premature termination phase.]  The adversary sends to $T$ the identities of the
parties that abort in the current round.
If there are less than $\badNum+1$ active parties,%
~then $T$ sends 
$\valueMultiParty{i-1}{J}$ to the active parties, where $\superSet{J}$ is the set of the active parties when parties can also abort during this phase (see exact details in \figref{MPCWithDealer}). The honest parties return this output and halt. 
\item[The main phase.]
If at least $\badNum+1$ parties are active, $T$ notifies the 
active parties that the protocol proceeds normally.
\end{description}
If after $\numRounds$ rounds, there are at least $\badNum+1$ active parties, $T$ sends 
$\outValue$ to all active parties and the honest parties output this value.

\begin{example}
As an example, assume that $\partNum=5$ and
$\badNum=3$. In this case the dealer computes a value $\valueMultiParty{i}{J}$ for every set of size 2 or 3.
Consider  an execution of the protocol where
$p_1$
aborts in round $4$ and $p_3$ and $p_4$ abort in round $100$. In this case,
$T$  sends $\valueMultiParty{99}{\set{2,5}}$ to $p_2$ and $p_5$, which
return  this output.
\end{example}


The formal proof of the $1/p$-security of
the protocol appears in \appref{OnlineProof}.
We next hint why for deterministic functionalities, any adversary can cause harm in the above protocol by
at most  $O(\domainSize^{O(1)}/\numRounds)$, where $\domainSize  = \domainSize(\secParam)$ is the
size of the domain of the inputs and the number of parties, i.e.,
$\partNum$, is constant. As in the protocols of~\cite{MNS09,GK10,BOO10},
the adversary can only cause harm by causing the protocol to
terminate in round $\istar$. In our protocol, if in some round there are
two values $\valueMultiParty{i}{J}$ and $\valueMultiParty{i}{J'}$ that the
adversary can obtain such that $\valueMultiParty{i}{J} \neq
\valueMultiParty{i}{J'}$, then the adversary can deduce that $i<\istar$.  Furthermore, the adversary might have some auxiliary information
on the inputs of the honest parties, thus, the adversary might be able to
deduce that a round is not $\istar$ even if all the values that it gets are
equal.  However, there are less than  $2^\badNum$ values that the adversary
can obtain in each round (i.e., the values of subsets of the $\badNum$
corrupt parties of size at least $\partNum-\badNum$).  We will show that
for  a round $i$ such that $i<\istar$, the probability  that all these
values are equal to a fixed value is $1/\domainSize^{O(1)}$ for a
deterministic function $f_\secParam$ (for a randomized functionality this
probability also depends on the size of the range).  By~\cite[Lemma 2]{GK10}, the protocol is
$\domainSize^{O(1)}/\numRounds$-secure.

\subsection{Eliminating the Dealer of the Protocol}
\label{sec:eliminatingDealer}

We eliminate the trusted on-line dealer in a few steps using a few layers
of secret-sharing schemes.  First, we change the  on-line dealer, so that,
in each round $i$, it shares the value $\valueMultiParty{i}{J}$ of each subset
$\superSet{J}$ among the parties of $\superSet{J}$ using a $|J|$-out-of-$|J|$ secret-sharing
scheme -- called \emph{inner} secret-sharing scheme.
As in Protocol $\MPCWithDealer_\numRounds$ described in \figref{MPCWithDealer}, the
adversary is able to obtain information on $\valueMultiParty{i}{J}$ only if it
controls all the parties in $\superSet{J}$. On the other hand,
the honest parties can reconstruct $\valueMultiParty{i-1}{J}$
(without the dealer),
where $\superSet{J}$ is the set of active parties containing the honest parties. 
In the reconstruction, if an active (corrupt) party does not give its share, then
it is removed from the set of active parties $\superSet{J}$. This is possible since in the case of a premature termination an honest majority among the active parties is guaranteed (as further explained below).

Next, we convert the on-line dealer to an off-line dealer. That is, we construct a protocol
in which the dealer sends only one message to each party in an
initialization stage; the parties interact in rounds using a broadcast
channel (without the  dealer) and in each round $i$ each party learns its
shares of the $i$th round inner secret-sharing schemes. In each round
$i$, each party $p_j$ learns a share of $\valueMultiParty{i}{J}$ in a
$|J|$-out-of-$|J|$ secret-sharing scheme, for every set $\superSet{J}$ such that
$j\in J$ and  $\partNum-\badNum \leq |J| \leq \badNum$
(that is, it learns the share of the inner scheme).  For this
purpose, the dealer computes, in a preprocessing  phase, the appropriate
shares for the inner secret-sharing scheme. For each round, the shares of 
each party $p_j$ are then shared in a $2$-out-of-$2$ secret-sharing scheme,
where $p_j$ gets one of the two shares (this share is a mask, enabling
$p_j$ to privately reconstruct its shares of the appropriate
$\valueMultiParty{i}{J}$ although messages are sent on a broadcast
channel). All other parties get shares in a $\badNum$-out-of-$(\partNum-1)$ Shamir
secret-sharing scheme  of the other share of the 2-out-of-2 secret-sharing.
See \constructref{shareWithRespect} for a formal description.
We call the resulting secret-sharing scheme the {\em outer} scheme.

To prevent corrupt parties from cheating, by say, sending false shares
and causing reconstruction of wrong secrets, every message that a party
should send during the  execution of the protocol is signed in the
preprocessing phase (together with the appropriate round number and  with
the party's index). In addition, the dealer sends a verification key to
each of  the parties.  To conclude, the off-line dealer gives each party
the signed shares for the outer secret  sharing scheme together with the
verification key.  A formal description of the functionality of the
off-line dealer, called Functionality $\MultiShareGenDomain$, is given in
\figref{MultiShareGenDomain}.

\myprotocolh{}{The initialization
functionality $\MultiShareGenDomain_\numRounds$.}{fig:MultiShareGenDomain}{    
\begin{enumerate}
\item[]{\hspace*{-0.63cm} \bf Joint input:} 
The security parameter $1^\secParam$, the number of rounds in the protocol $\numRounds=\numRounds(\secParam)$,
                a bound $\badNum$ on the number of corrupted parties, and the set of indices of aborted parties $\abortedAfterPre$.
                
\item[]{\hspace*{-0.63cm} \bf Private input:}
                Each party $p_j$, where $j\notin \abortedAfterPre$, has an input $\partyInput{j}\in \partyDomain{}$.

\item[]{\hspace*{-0.63cm} \bf Computing default values and signing keys} 
        \item \label{stp:defaultInput} For every $j\in \abortedAfterPre$, select $\partyInput{j}$ 
        with uniform distribution from $\partyDomain{}$.
        \item\label{stp:selectIStar} Select $\istar \in \range{\numRounds}$ with uniform distribution and compute $\outValue \gets  f_\secParam(\partyInput{1},\ldots,\partyInput{\partNum})$.
                                        
                                
\item \label{stp:beforeIStar}For each $1\le i < \istar$, for each $J \subseteq
\range{\partNum}\setminus \abortedAfterPre$ s.t. $\goodNum \leq \size{J} \leq \badNum$,
\begin{enumerate}
        \item For each $j\in J$, set $\randomDomain{j}{i} = \partyInput{j}$.
        \item For each $j\not\in J$, select uniformly at random $\randomDomain{j}{i} \in \partyDomain{}$.
        \item Set $\valueMultiParty{i}{J} \gets
f_\secParam(\randomDomain{1}{i},\ldots,\randomDomain{\partNum}{i})$.
\end{enumerate}

  \item \label{stp:chooseValues2} For each $\istar\le i \le \numRounds$ and 
  for each $J \subseteq \range{\partNum}\setminus \abortedAfterPre$ s.t. $\goodNum \leq \size{J} \leq \badNum$,
                                        set $\valueMultiParty{i}{J} = \outValue$.

   \item \label{stp:produceKeys1} Compute $\vect{K_{\rm sign},K_{\rm ver}} \leftarrow \operatorname{Gen}(1^n)$. 
\item[] {\hspace*{-0.63cm} \bf Computing signed shares of the inner secret-sharing scheme}

\item
\label{stp:inner}
  For each $i \in \set{1, \ldots, \numRounds}$ and 
  for each $J \subseteq \range{\partNum} \setminus \abortedAfterPre$ s.t. $\goodNum \leq \size{J} \leq \badNum$,
  \begin{enumerate}
  \item Create shares of $\valueMultiParty{i}{J}$ in a
   $|J|$-out-of-$|J|$  secret-sharing scheme for the parties in
   $\superSet{J}$.  For each party $p_j \in \superSet{J}$, let
   $\valueInnerSecretSharingShare{i}{J}{j}$ be its share of
   $\valueMultiParty{i}{J}$.
   \item Sign each share $\valueInnerSecretSharingShare{i}{J}{j}$: 
   compute $\valueInnerSecretSharingShareSigned{i}{J}{j}\leftarrow
   (\valueInnerSecretSharingShare{i}{J}{j}, i, J, j ,
   \operatorname{Sign}((\valueInnerSecretSharingShare{i}{J}{j}, i, J, j) ,K_{\rm sign})).$
 \end{enumerate}
\item[] {\hspace*{-0.63cm} \bf Computing shares of the outer secret-sharing scheme}
\item \label{stp:outer} For each $i \in \range{\numRounds}$, 
for each $J \subseteq \range{\partNum} \setminus \abortedAfterPre$ s.t. $\goodNum \leq \size{J} \leq \badNum$,
and each $j \in J$,
share $\valueInnerSecretSharingShareSigned{i}{J}{j}$ 
using a $(\badNum + 1)$-out-of-$\partNum$ secret-sharing scheme 
with respect to $p_{j}$ as defined in \constructref{shareWithRespect}:
compute 
one masking share $\masking{\valueInnerSecretSharingShareSigned{i}{J}{j}}{j}$ and $\partNum - 1$ complement shares
                   $\vect{\complParty{\valueInnerSecretSharingShareSigned{i}{J}{j}}{1},\ldots,\complParty{\valueInnerSecretSharingShareSigned{i}{J}{j}}{j-1},\complParty{\valueInnerSecretSharingShareSigned{i}{J}{j}}{j+1},\ldots,\complParty{\valueInnerSecretSharingShareSigned{i}{J}{j}}{\partNum}}$.
\item[] {\hspace*{-0.63cm} \bf Signing the messages of all parties}
\item \label{stp:preperMSG} For every $1\leq q \leq \partNum$, compute the message $m_{q,i}$ that $p_q\in P$ broadcasts in round $i$
  by concatenating (1) $q$, (2) $i$, 
  and (3) the complement shares $\complParty{\valueInnerSecretSharingShareSigned{i}{J}{j}}{q}$ produced in \stepref{outer} for $p_q$ (for all $J \subseteq \range{\partNum}\setminus \abortedAfterPre$ s.t. $\goodNum \leq \size{J} \leq \badNum$ and all $j\neq q$ s.t. $j \in J$), and
  compute $M_{q,i} \leftarrow 
  \vect{m_{q,i},\operatorname{Sign}(m_{q,i},K_{\rm sign})}$.
\item[]{\hspace*{-0.63cm} \bf Outputs:}  Each party $p_j$ such that $j\notin \abortedAfterPre$ receives 
\begin{itemize}
   \item The verification key $K_{\rm ver}$.
   \item The messages $M_{j,1}, \ldots, M_{j,\numRounds}$ that $p_j$ 
           broadcasts during the  protocol.
   \item $p_j$'s private masks $\masking{\valueInnerSecretSharingShareSigned{i}{J}{j}}{j}$
produced in \stepref{outer}, for each $1\leq i \leq \numRounds$ and each $J \subseteq \range{\partNum} \setminus \abortedAfterPre$ s.t. $\goodNum \leq \size{J} \leq \badNum$ and $j \in J$.
 \end{itemize}
\end{enumerate}
}

The protocol with the off-line dealer proceeds in rounds.  In round $i$ of
the protocol all parties broadcast their (signed) shares  in the outer
($\badNum+1$)-out-of-$\partNum$ secret-sharing scheme. Thereafter, each
party can  unmask the message it receives (with its share in the
appropriate $2$-out-of-$2$  secret-sharing scheme) to obtain its shares in
the $|J|$-out-of-$|J|$ inner secret-sharing of the values $\valueMultiParty{i}{J}$ (for the
appropriate sets $\superSet{J}$'s to which the party belongs). If a party stops
broadcasting messages  or broadcasts improperly signs messages, then all
other parties consider it as  aborted. If $\partNum-\badNum $ or more parties abort, the
remaining parties reconstruct the value of the set
that contains all of them, i.e., $\valueMultiParty{i-1}{J}$. In the special case
of premature termination already in the first round, the remaining active parties
engage in a fully secure protocol (with honest majority) to compute $f_\secParam$.

The use of the outer secret-sharing scheme with threshold $\badNum+1$ plays
a crucial role in eliminating the on-line dealer. On the one hand, it
guarantees that an adversary, corrupting at most $\badNum$ parties, cannot
reconstruct the shares of round $i$  before round $i$.  On the other hand,
at least $\partNum-\badNum$ parties must abort to prevent the reconstruction
of the outer secret-sharing scheme (this is why we cannot proceed after $\partNum-\badNum$
parties aborted).
Furthermore, since $\badNum \leq 2\partNum/3$, when at least $\goodNum$
corrupt parties aborted, there is an honest majority.
To see this, assume that at least $\goodNum$ corrupt parties aborted.
Thus, at most $\badNum - (\goodNum) = 2\badNum - \partNum$ corrupt parties are active.
There are $\goodNum$ honest parties (which are obviously active), therefore, as
$2\badNum - \partNum < \goodNum$ (since $\badNum< 2\partNum / 3$), an honest majority is achieved when $\goodNum$ parties abort.
In this case we can execute a protocol with full security for the reconstruction.

Finally, we replace the off-line dealer by using a secure-with-abort and
cheat-detection protocol computing the functionality computed by the
dealer, that is, Functionality $\MultiShareGenDomain_\numRounds$. Obtaining
the outputs of this computation, an adversary is unable to infer any
information regarding the input of honest parties or the output of the
protocol (since it gets $\badNum$ shares of a $(\badNum+1)$-out-of-$\partNum$ secret-sharing scheme).
The adversary, however, can prevent the execution, at the price
of at least one corrupt party being detected cheating by all other parties. In
such an event, the remaining  parties will start over without the detected cheating party. This goes on either
until the protocol succeeds or there is an honest majority and a fully secure
protocol computing $f_\secParam$ is executed.  

A formal description of the protocol
appears in \figref{MultiPartyMPC}.  The reconstruction functionality used
in this protocol (when at least $\partNum-\badNum$ parties aborted) appears
in \figref{Reconstruction}.
The details of how to construct
a protocol secure-with-abort and cheat-detection with $O(1)$ rounds are given in~\cite{BOO10}.

\myprotocolh{}
{The $\partNum$-party protocol $\MPC_\numRounds$ for computing $\F$.}
{fig:MultiPartyMPC}{
\begin{enumerate}
\item[]{\hspace*{-0.63cm} \bf Inputs:} 
        Each party $p_j$ holds the private input $\initialInput{j} \in \partyDomain{}$ and the joint input:     
  the security parameter $1^\secParam$, the number of rounds in the protocol $\numRounds=\numRounds(\secParam)$,
                and a bound $\badNum$ on the number of corrupted parties.
\item[]\hspace{-0.63cm} {\bf Preliminary phase:}
\item $\abortedAfterPre = \emptyset$
\item \label{stp:executeShareGenWithAbortManyTimes} If $\size{\abortedAfterPre} < \goodNum$, 
\begin{enumerate}
        \item \label{stp:executeShareGenWithAbort} The parties in $\set{p_j : j\in \range{m} \setminus \abortedAfterPre}$ execute a secure-with-abort and cheat-detection  protocol computing Functionality  $\MultiShareGenDomain_\numRounds$. Each honest party $p_j$ inputs $\initialInput{j}$ as its input for the functionality.
        \item \label{stp:abortShareGenWithAbort}
        If a party $p_j$ aborts, that is, the output of the honest parties
        is $\abort{j}$, then, set $\abortedAfterPre = \abortedAfterPre \cup \set{j}$,
        chose $\partyInput{j}$ uniformly at random from $\partyInput{j}$, and goto \stepref{executeShareGenWithAbortManyTimes}.
        \item   Else (no party has aborted), denote $\abortedI =\abortedAfterPre$ and proceed to the first round.
                                               
\end{enumerate}

\item \label{stp:fairMPC}   Otherwise ($\size{\abortedAfterPre} \ge \goodNum$), the premature termination is executed with $i=1$.
\item[]\hspace{-0.63cm} {\bf  In each round $i=1,\ldots, \numRounds$ do:}
\item \label{stp:broadcastMessage} Each party $p_j$ broadcasts $M_{j,i}$ (containing
                its shares in the outer secret-sharing scheme).

\item \label{stp:broadcastBadMessage} For every $p_j$ s.t. $\operatorname{Ver}(M_{j,i}, K_{\rm ver})=0$ or if $p_j$ broadcasts 
an invalid or no message, then all parties mark $p_j$ as inactive, i.e., set 
$\abortedI  = \abortedI  \cup \set{j}$.
If $\size{\abortedI } \geq \goodNum$, premature termination is executed.
\item[]\hspace{-0.63cm} {\bf  Premature termination step}
\item \label{stp:prematureTerminationFirstRound} If $i=1$, 
                        the active parties use a multiparty secure protocol (with full security) 
        to compute $f_\secParam$: Each honest party inputs $\initialInput{j}$ and the input of each inactive party is chosen uniformly at random from $\partyDomain{}$. The active parties output the result, and halt. 
\item \label{stp:prematureTerminationNotFirstRound}
Otherwise,
\begin{enumerate}
\item \label{stp:sendMSGPrematureTermination} Each  party $p_j$ reconstructs $\valueInnerSecretSharingShareSigned{i-1}{J}{j}$, the signed share of the inner secret-sharing scheme
  produced in \stepref{inner} of Functionality $\MultiShareGenDomain_\numRounds$,  
  for each $J \subseteq \range{\partNum}\setminus \abortedAfterPre$ s.t. $\goodNum \leq \size{J} \leq \badNum$ and $j\in J$.
\item \label{stp:executeReconstruction} The active parties execute a secure multiparty  protocol with an
honest majority to compute Functionality $\Reconstruction$, where the input of each
party $p_j$ is $\valueInnerSecretSharingShareSigned{i-1}{J}{j}$
for every $J \subseteq \range{\partNum}\setminus \abortedAfterPre$ s.t. $\goodNum \leq \size{J} \leq \badNum$
and $j \in J$.
\item \label{stp:outputReconstruction} The active parties output the output of this protocol, and halt.
\end{enumerate}
\item[]\hspace{-0.63cm} {\bf  At the end of round $\numRounds$:}
\item \label{stp:lastRoundBroadcast} Each active party $p_j$ broadcasts the signed shares
  $\valueInnerSecretSharingShareSigned{\numRounds}{J}{j}$ for each $J$ such that  $j \in J$.
\item \label{stp:lastRoundOutput} Let $J\subseteq \range{\partNum} \setminus \abortedI $ be the lexicographical first set such that all the parties in $\superSet{J}$ broadcast properly signed shares 
                                $\valueInnerSecretSharingShareSigned{\numRounds}{J}{j}$.
        Each active party reconstructs the value $\valueMultiParty{\numRounds}{J}$, outputs $\valueMultiParty{\numRounds}{J}$, and halts.
\end{enumerate}
}

\myprotocolh{}{Functionality $\Reconstruction$
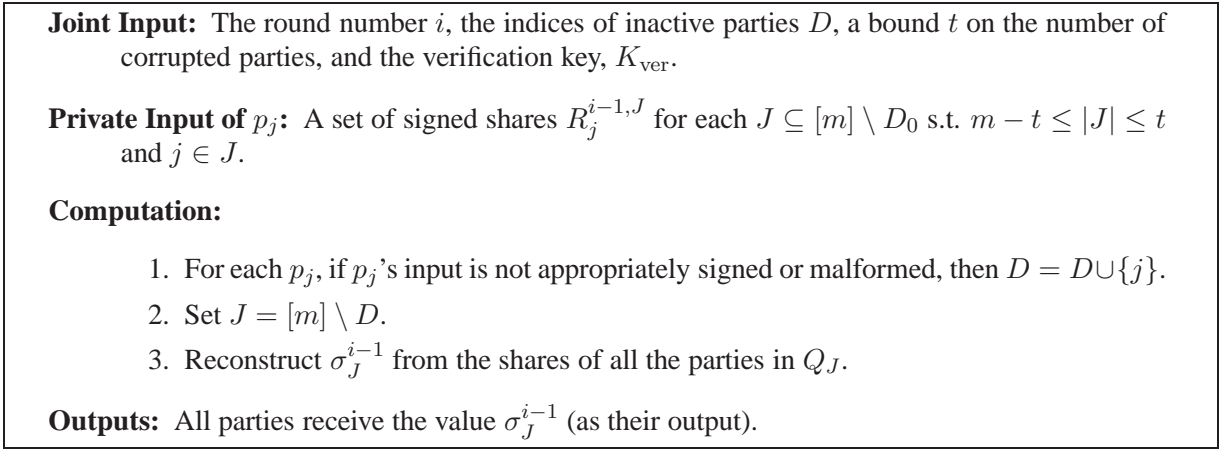
for reconstructing the output in the premature termination step.}{fig:Reconstruction}{
\begin{description}
\item[\bf Joint Input:] The round number $i$, the indices of inactive parties $\abortedI $, a bound $\badNum$ on the number of corrupted parties,
 and the verification key, $K_{\rm ver}$.
\item[\bf Private Input of $p_{j}$:]
   A set of signed shares $\valueInnerSecretSharingShareSigned{i-1}{J}{j}$  
   for each $J \subseteq \range{\partNum}\setminus \abortedAfterPre$ s.t. $\goodNum \leq \size{J} \leq \badNum$ and $j\in J$.
\item[\bf Computation:] ~ 
\begin{enumerate}
  \item For each $p_{j}$, if $p_{j}$'s input is not appropriately signed or malformed,
                        then $\abortedI  = \abortedI  \cup \set{j}$.
  \item Set $J = \range{\partNum} \setminus \abortedI $.
\item Reconstruct $\IOLvalue_J^{i-1}$ from the shares of all the parties in $\superSet{J}$. 
%
\end{enumerate}

\item[\bf Outputs:]
 All parties receive the value $\IOLvalue_J^{i-1}$ (as their output).
\end{description}
}

\paragraph{Comparison with the multiparty coin-tossing protocol
  of~\cite{BOO10}.}  Our protocol combines ideas from the protocols
of~\cite{GK10,BOO10}. However, there are some important differences between
our protocol and the protocol of~\cite{BOO10}. In the coin-tossing protocol
of~\cite{BOO10}, the bits $\valueMultiParty{i}{J}$ are shared using a
threshold scheme where the threshold is smaller than the size of the set
$\superSet{J}$. This means that a proper subset of $\superSet{J}$ containing
corrupt parties can reconstruct $\valueMultiParty{i}{J}$. In coin-tossing
this is not a problem since there are no inputs. However, when computing
functionalities with inputs, such $\valueMultiParty{i}{J}$ might reveal
information on the inputs of honest parties in $\superSet{J}$, and we share
$\valueMultiParty{i}{J}$ with threshold $\size{\superSet{J}}$. As a result, we use
more sets $\superSet{J}$ than in~\cite{BOO10} and the bias of the protocol
is increased (put differently, to keep the same security, we need to increase the number of rounds in the protocol). For example, the protocol of~\cite{BOO10} has small bias when there are polynomially many parties and $\badNum=\partNum/2$. Our protocol is efficient only when there are constant number of parties. As explained
in \secref{Impossibility}, this difference is inherent as a protocol
for general functionalities with polynomially many parties and $\badNum=\partNum/2$ cannot have a small bias.

\subsection{A $1/p$-Secure Protocol for Polynomial Range}
\label{sec:polynomialRange}
Using an idea of~\cite{GK10}, we modify our protocol such that it will
have a small bias when the size of the range of the functionality
$\F$ is polynomially bounded (even if $\F$ is randomized
and has a big domain of inputs). The only modification is the way that each
$\valueMultiParty{i}{J}$ is chosen prior to round $\istar$: with
probability $1/(2p)$ we choose $\valueMultiParty{i}{J}$ as a random value
in the range of $f_\secParam$ and with probability $1-1/(2p)$ we choose it as
in \figref{MultiShareGenDomain}. 
Formally, in the model with the dealer, in the preprocessing phase of $\MPCWithDealer_\numRounds$
described in \figref{MPCWithDealer}, we replace \stepref{beforeIStar1} with the
following step:
\begin{itemize} 
        \item For each $i \in \set{1, \ldots, \istar-1}$ and for each $J \subseteq \range{\partNum} \setminus \abortedAfterPre$ s.t. $\goodNum \leq \size{J} \leq \badNum$,
         \begin{itemize}  
                                                \item with probability $1/(2p)$, select uniformly at random $z_J^i \in \partyRange{}$ and set $\valueMultiParty{i}{J} = z_J^i$.
                                                \item   with the remaining probability $1-1/(2p)$,
                                                                                \begin{enumerate}
                                                                                        \item For every $j\not \in J$ select uniformly at random $\randomDomain{j}{i} \in \partyDomain{}$ and for each $j\in J$, set $\randomDomain{j}{i} = \partyInput{j}$.
                                                                                        \item Compute $\valueMultiParty{i}{J} \gets 
                                                                                        f_\secParam(\randomDomain{1}{i},\ldots,\randomDomain{\partNum}{i})$.
                                                                                \end{enumerate}
                                                                              \end{itemize}
        
\end{itemize}
Similarly, in the protocol without the dealer, Protocol $\MPC_\numRounds$, we replace
\stepref{beforeIStar} in $\MultiShareGenDomain_\numRounds$
(described in \figref{MultiShareGenDomain}) with the above step.  Denote the resulting protocols
with and without the dealer models by $\MPCWithDealerRange$ and
$\MPCRange_\numRounds$, respectively.

The idea why this change improves the protocol is that now the probability
that all values held by the adversary are equal prior to round $\istar$
is bigger, thus, the probability that the adversary guesses $\istar$ is
smaller. This modification, however, can cause the honest parties to output
a value that is not possible given their inputs, and, in general, we cannot
simulate the case (which happens with probability $1/(2p)$) when the output is chosen with uniform distribution
from the range.

\section{Impossibility of $1/p$-secure Computation with Non-Constant Number
  of Parties}
\label{sec:Impossibility}

For deterministic functions, our protocol is efficient when the number of
parties $\partNum$ is constant and the size of the domain or range is polynomial (in the
security parameter $\secParam$) or when the number of parties is 
$O(\log \log \secParam)$ and the size of the domain is constant.
We next show that, in
general, there is no efficient protocol when the number of parties is
$\partNum(\secParam)=\omega(1)$ and the size of the domain is polynomial
and when   $\partNum(\secParam)=\omega(\log \secParam)$ and the size of the
domain of each party is 2. This is done using the  following impossibility
result of Gordon and Katz~\cite{GK10}.
\begin{theorem}[\cite{GK10}]
\label{thm:GK}
  For every $\ell(\secParam)=\omega(\log \secParam)$, there exists a
  deterministic 2-party functionality  $\F$ with domain and range
  $\set{0,1}^{\ell(\secParam)}$ that cannot be $1/p$-securely
  computed for $p\geq 2+1/\poly(\secParam)$.
\end{theorem}

We next state and prove our impossibility results.

\begin{theorem}
\label{thm:impossibility}
   For every $\partNum(\secParam)=\omega(\log \secParam)$, there exists a
   deterministic $\partNum(\secParam)$-party functionality $\F'$ with domain
   $\set{0,1}$ that cannot be $1/p$-securely computed for $p\geq
   2+1/\poly(\secParam)$ without an honest majority.
\end{theorem}

\begin{proof}
Let $\ell(\secParam) = \partNum(\secParam)/2$ (for simplicity, assume $\partNum(\secParam)$ is even).
Let $\F=\set{f_\secParam}_{\secParam\in \NN}$ be the functionality
guaranteed in \thmref{GK} for $\ell(\secParam)$.
Define an $\partNum(\secParam)$-party deterministic functionality
$\F'=\set{f'_\secParam}_{\secParam\in \NN}$, where in $f'_\secParam$ party
$p_j$ gets the $j$th bit of the inputs of $f_\secParam$ and the outputs of
$f_\secParam$ and $f'_\secParam$ are equal 
Assume that $\F'$ can be 
$1/p$-securely computed by a protocol $\Pi'$ assuming that
$\badNum(\secParam)=\partNum(\secParam)/2$ parties can be corrupted. This implies a $1/p$-secure protocol $\Pi$ for $\F$ with two parties, where the first party simulates the first $\badNum(\secParam)$ parties in $\Pi'$ and the second party simulates the last $\badNum(\secParam)$ parties. The $1/p$-security of $\Pi$ is implied by the fact that any adversary $\A$ for the protocol $\Pi$ can be transformed into an adversary $\A'$ for $\Pi'$ controlling
$\partNum(\secParam)/2 = \badNum(\secParam)$ parties; as $\A'$ cannot violate the $1/p$-security of $\Pi'$, the adversary $\A$ cannot violate the $1/p$-security of $\Pi$.
\end{proof}

\begin{theorem}
\label{thm:imposibility2}
  For every $\partNum(\secParam)=\omega(1)$, there exists a deterministic
  $\partNum(\secParam)$-party functionality $\F''$ with domain
  $\set{0,1}^{\log \secParam}$ that cannot be $1/p$-securely computed for
  $p\geq 2+1/\poly(\secParam)$ without an honest majority.
\end{theorem}
\begin{proof}
  Let $\ell(\partNum)=0.5\partNum(\secParam)\log \secParam$ and let
$\F=\set{f_\secParam}_{\secParam\in \NN}$ be the functionality guaranteed
in \thmref{GK} for $\ell(\partNum)$.  We divide the $2\ell(\secParam)$ bits
of the inputs of $f_\secParam$ into $\partNum(\secParam)$ blocks of length $\log
\secParam$.  Define an $\partNum(\secParam)$-party deterministic functionality
$\F''=\set{f''_\secParam}_{\secParam\in \NN}$, where in $f''_\secParam$
party $p_j$ gets the $j$th block of the inputs of $f_\secParam$ and the outputs
of $f_\secParam$ and $f''_\secParam$ are equal.  As in the proof of
\thmref{impossibility}, a $1/p$-secure protocol for $\F''$ implies a
$1/p$-secure protocol for $\F$ contradicting \thmref{GK}.
\end{proof}

\medskip

The above impossibility results should be contrasted with the
coin-tossing protocol of~\cite{BOO10} which is an efficient $1/p$-secure
protocol even when $\partNum$ is polynomial in the security parameter and
the number of bad parties is $\partNum(\secParam)/2+O(1)$.  Notice that in
both our impossibility results the size of the range is
super-polynomial (as we consider the model where all parties get the same output).
It is open if there is an efficient $1/p$-secure protocol
when the number of parties is not constant and the size of both the domain
and range is polynomial.

\bibliographystyle{plain}
\bibliography{bib-read-only}

\begin{thebibliography}{10}

\bibitem{AL07}
Y.~Aumann and Y.~Lindell.
\newblock Security against covert adversaries: Efficient protocols for
  realistic adversaries.
\newblock In S.~Vadhan, editor, {\em Proc. of the Fourth Theory of Cryptography
  Conference -- TCC 2006}, volume 4392 of {\em Lecture Notes in Computer
  Science}, pages 137--156. Springer-Verlag, 2007.

\bibitem{BG89}
D.~Beaver and S.~Goldwasser.
\newblock Multiparty computation with faulty majority.
\newblock In {\em Proc. of the 30th IEEE Symp. on Foundations of Computer
  Science}, pages 468--473, 1989.

\bibitem{BOO10FULL}
A.~Beimel, E.~Omri, and I.~Orlov.
\newblock Protocols for multiparty coin toss with dishonest majority.
\newblock Full version of~\cite{BOO10}.

\bibitem{BOO10}
A.~Beimel, E.~Omri, and I.~Orlov.
\newblock Protocols for multiparty coin toss with dishonest majority.
\newblock In T.~Rabin, editor, {\em Advances in Cryptology -- CRYPTO 2010},
  volume 6223 of {\em Lecture Notes in Computer Science}, pages 538--557.
  Springer-Verlag, 2010.

\bibitem{BGMR90}
M.~Ben-Or, O.~Goldreich, S.~Micali, and R.~Rivest.
\newblock A fair protocol for signing contracts.
\newblock In {\em Proceedings of the 12th Colloquium on Automata, Languages and
  Programming}, pages 43--52. Springer-Verlag, 1985.

\bibitem{Bl84}
M.~Blum.
\newblock How to exchange (secret) keys.
\newblock {\em ACM Trans. Comput. Syst.}, 1(2):175--193, 1983.

\bibitem{BN00}
D.~Boneh and M.~Naor.
\newblock Timed commitments.
\newblock In M.~Bellare, editor, {\em Advances in Cryptology -- CRYPTO 2000},
  volume 1880 of {\em Lecture Notes in Computer Science}, pages 236--254.
  Springer-Verlag, 2000.

\bibitem{Can00}
R.~Canetti.
\newblock Security and composition of multiparty cryptographic protocols.
\newblock {\em J. of Cryptology}, 13(1):143--202, 2000.

\bibitem{Cle86}
R.~Cleve.
\newblock Limits on the security of coin flips when half the processors are
  faulty.
\newblock In {\em Proc. of the 18th STOC}, pages 364--369, 1986.

\bibitem{Cle90}
R.~Cleve.
\newblock Controlled gradual disclosure schemes for random bits and their
  applications.
\newblock In G.~Brassard, editor, {\em Advances in Cryptology -- CRYPTO '89},
  volume 435 of {\em Lecture Notes in Computer Science}, pages 573--588.
  Springer-Verlag, 1990.

\bibitem{Dam95}
I.~Damg{\aa}rd.
\newblock Practical and provably secure release of a secret and exchange of
  signatures.
\newblock {\em J. of Cryptology}, 8(4):201--222, 1995.

\bibitem{EGL85}
S.~Even, O.~Goldreich, and A.~Lempel.
\newblock A randomized protocol for signing contracts.
\newblock {\em CACM}, 28(6):637--647, 1985.

\bibitem{GHY87}
Z.~Galil, S.~Haber, and M.~Yung.
\newblock Cryptographic computation: Secure fault-tolerant protocols and the
  public-key model.
\newblock In C.~Pomerance, editor, {\em Advances in Cryptology -- CRYPTO '87},
  volume 293 of {\em Lecture Notes in Computer Science}, pages 135--155.
  Springer-Verlag, 1988.

\bibitem{GMPY06}
J.~A. Garay, P.~D. MacKenzie, M.~Prabhakaran, and K.~Yang.
\newblock Resource fairness and composability of cryptographic protocols.
\newblock In S.~Halevi and T.~Rabin, editors, {\em Proc. of the Third Theory of
  Cryptography Conference -- TCC 2006}, volume 3876 of {\em Lecture Notes in
  Computer Science}, pages 404--428. Springer-Verlag, 2006.

\bibitem{Gol04}
O.~Goldreich.
\newblock {\em Foundations of Cryptography, Voume II Basic Applications}.
\newblock Cambridge University Press, 2004.

\bibitem{GMW87}
O.~Goldreich, S.~Micali, and A.~Wigderson.
\newblock How to play any mental game.
\newblock In {\em Proc. of the 19th ACM Symp. on the Theory of Computing},
  pages 218--229, 1987.

\bibitem{GL90}
S.~Goldwasser and L.~Levin.
\newblock Fair computation of general functions in presence of immoral
  majority.
\newblock In A.~J. Menezes and S.~A. Vanstone, editors, {\em Advances in
  Cryptology -- CRYPTO '90}, volume 537 of {\em Lecture Notes in Computer
  Science}, pages 77--93. Springer-Verlag, 1991.

\bibitem{GL02}
S.~Goldwasser and Y.~Lindell.
\newblock Secure computation without agreement.
\newblock In {\em DISC '02: Proceedings of the 16th International Conference on
  Distributed Computing}, pages 17--32, London, UK, 2002. Springer-Verlag.

\bibitem{GK09}
D.~Gordon and J.~Katz.
\newblock Complete fairness in multi-party computation without an honest
  majority.
\newblock In {\em Proc. of the Sixth Theory of Cryptography Conference -- TCC
  2009}, pages 19--35, Berlin, Heidelberg, 2009. Springer-Verlag.

\bibitem{GHKL08}
S.~D. Gordon, C.~Hazay, J.~Katz, and Y.~Lindell.
\newblock Complete fairness in secure two-party computation.
\newblock In {\em Proc. of the 40th ACM Symp. on the Theory of Computing},
  pages 413--422, 2008.

\bibitem{GIMOS10}
S.~D. Gordon, Y.~Ishai, T.~Moran, R.~Ostrovsky, and A.~Sahai.
\newblock On complete primitives for fairness.
\newblock In D.~Micciancio, editor, {\em Proc. of the Seventh Theory of
  Cryptography Conference -- TCC 2010}, volume 5978 of {\em Lecture Notes in
  Computer Science}, pages 91--108. Springer-Verlag, 2010.

\bibitem{GK10}
S.~D. Gordon and J.~Katz.
\newblock Partial fairness in secure two-party computation.
\newblock In Henri Gilbert, editor, {\em Advances in Cryptology -- EUROCRYPT
  2010}, volume 6110 of {\em Lecture Notes in Computer Science}, pages
  157--176. Springer-Verlag, 2010.

\bibitem{LMR83}
M.~Luby, S.~Micali, and C.~Rackoff.
\newblock How to simultaneously exchange a secret bit by flipping a
  symmetrically-biased coin.
\newblock In {\em Proc. of the 24th IEEE Symp. on Foundations of Computer
  Science}, pages 11--21, 1983.

\bibitem{MNS09}
T.~Moran, M.~Naor, and G.~Segev.
\newblock An optimally fair coin toss.
\newblock In {\em Proc. of the Sixth Theory of Cryptography Conference -- TCC
  2009}, pages 1--18, 2009.

\bibitem{Pin03}
B.~Pinkas.
\newblock Fair secure two-party computation.
\newblock In E.~Biham, editor, {\em Advances in Cryptology -- EUROCRYPT 2003},
  volume 2656 of {\em Lecture Notes in Computer Science}, pages 87--105.
  Springer-Verlag, 2003.

\bibitem{Sh}
A.~Shamir.
\newblock How to share a secret.
\newblock {\em Communications of the ACM}, 22:612--613, 1979.

\bibitem{Yao86}
A.~C. Yao.
\newblock How to generate and exchange secrets.
\newblock In {\em Proc. of the 27th IEEE Symp. on Foundations of Computer
  Science}, pages 162--167, 1986.

\end{thebibliography}

\appendix

\section{Security with Abort and Cheat Detection}
\label{app:securityWithBlaming}

We next present a definition of secure multiparty computation that is more 
stringent than standard definitions of secure computation with abort. This 
definition extends the definition for secure computation as given by Aumann 
and Lindell~\cite{AL07}. Roughly speaking, the definition requires that one of two events is 
possible:
(1) The protocol terminates normally, and \emph{all} parties receive 
their outputs, or
(2) Corrupted parties deviate from the prescribed protocol; in this case the adversary obtains the outputs of the corrupted parties (but nothing 
else), and all honest parties are given an
identity of one party that has aborted. 
The formal definition uses the real vs.~ideal paradigm as discussed in \secref{realIdeal}. 
We next describe the appropriate ideal model.

\paragraph{Execution in the ideal model.} 
Let 
$\badI \subseteq \range{\partNum}$ denote the set of indices of corrupted parties
controlled by an adversary $\A$. 
The adversary $\A$ receives an auxiliary input denoted $\aux$.
An ideal execution proceeds as follows:
\begin{description}
        \item{\textbf{Send inputs to trusted party:}} The honest parties
                send their inputs to the trusted party.  
                The corrupted parties may either send their received input, or send some other input of the same
                length (i.e., $\partyInput{j} \in \partyDomain{}$) to the trusted party,
                or abort (by sending a special $\abort{j}$ message).            
                Denote by
                $\partyInput{1},\ldots,\partyInput{\partNum}$ the
                inputs received by the trusted party. If the trusted party receives
                an $\abort{j}$ message, then it sends $\abort{j}$ to all 
                honest parties and terminates (if it received $\abort{j}$ from more 
                than one $j$, then it uses the minimal such $j$).  

        \item{\textbf{Trusted party sends outputs to adversary:}} The
                trusted party computes $\outValue \gets
                f_\secParam(\partyInput{1},\ldots,\partyInput{\partNum})$
                and sends the output $\outValue$ to the adversary.
                
        \item{\textbf{Adversary instructs the trusted party to continue or
                halt:}} $\A$ sends either a $\continueMSG$ message or $\abort{j}$
                to the trusted party  for some corrupt party $p_j$, i.e., $j\in\badI$.
                If it sends a $\continueMSG$ message, the trusted party sends $\outValue$ to 
                all honest parties. Otherwise, if the adversary sends
                $\abort{j}$, then the trusted party sends $\abort{j}$  to all
                honest parties.
        \item{\textbf{Outputs:}} An honest party always outputs the value
                $\outValue$ it obtained from the trusted party. The
                corrupted parties output nothing.  The adversary $\A$
                outputs any (probabilistic polynomial-time  computable)
                function of the auxiliary input $\aux$, 
                the inputs of the corrupt parties,
                and the value $\outValue$ obtained from the trusted party.
\end{description}

We let $\Ideal_{\F,\S(\aux)}^{\rm{CD}}(\vecInput,1^\secParam)$ and
$\Real_{\Pi,A(\aux)}(\vecInput,1^\secParam)$ be defined as in
\secref{realIdeal} (where in this case $\Ideal_{\F,\S(\aux)}^{\rm{CD}}(\vecInput,1^\secParam)$  refers to the above
execution with cheat-detection of $\F$).  This ideal model is different
from that of~\cite{Gol04} in  that in the case of an ``abort'', the honest
parties get output $\abort{j}$  and not a $\bot$ symbol. This means that the
honest parties \emph{know}  an identity of a corrupted party that causes
the abort.  This cheat-detection is achieved by most multiparty protocols,
including  that of~\cite{GMW87}, but not all (e.g., the protocol
of~\cite{GL02} does not  meet this requirement).  Using this notation we
define secure computation with abort and cheat-detection.

\begin{definition}[security-with-abort and cheat-detection]\label{def:cheatDetection} 
Let $\F$ and $\Pi$ be as in \defref{1Overp-security}. A protocol $\Pi$ is said to 
\texttt{securely compute} $\F$ against at most $\badNum(\secParam)$ corrupt parties
with abort and cheat-detection if for 
every non-uniform polynomial-time adversary $\A$ in the real model
controlling at most $\badNum(\secParam)$ parties, 
there exists a non-uniform polynomial-time adversary $\S$ in the ideal 
model controlling the same parties, such that 
$$\set{\Ideal_{\F,\S(\aux)}^{\rm{CD}}(\vecInput,1^\secParam)}_{
  \aux \in \set{0,1}^*,\vecInput \in (\partyDomain)^\partNum,\secParam\in\NN} 
\quad\comp\quad
  \set{\Real_{\Pi,A(\aux)}(\vecInput,1^\secParam)}_{\aux \in \set{0,1}^*,\vecInput \in (\partyDomain)^\partNum,\secParam\in\NN}.$$
\end{definition}

\section{Cryptographic Tools}
\label{app:cryptoTools}
\paragraph{Signature Schemes.} 
Informally, a signature on a message proves that the
message was created by its presumed sender, and its content was not altered.  A
signature scheme is  a triple
$\vect{\operatorname{Gen},\operatorname{Sign},\operatorname{Ver}}$ containing 
the key generation algorithm $\operatorname{Gen}$, which 
outputs a  pair of keys,  the signing key $K_{S}$ and the
verification key $K_{v}$,  the signing algorithm $\operatorname{Sign}$, and the verifying algorithm $\operatorname{Ver}$. We assume that it is infeasible to produce signatures without holding the signing key.
For formal definition see~\cite{Gol04}.

\paragraph{Secret Sharing Schemes.}\label{sec:secretSharing}
An $\alpha$-out-of-$\partNum$ secret-sharing scheme is a mechanism for sharing data among a set of parties such that every set of size $\alpha$ can reconstruct the secret, while any smaller  set knows nothing about the secret.
In this paper, we use two schemes: the XOR-based $\partNum$-out-of-$\partNum$ scheme 
(i.e., in this scheme $\alpha = \partNum$)
and
Shamir's $\alpha$-out-of-$\partNum$ secret-sharing scheme~\cite{Sh} 
which is used when $\alpha < \partNum$.
In both schemes, for
every $\alpha-1$ parties, the shares of these parties are uniformly
distributed and independent of the secret.
Furthermore, given such $\alpha-1$ shares and a secret $s$, one can {\em efficiently} complete them to $\partNum$ shares of the secret $s$.

In our protocols we sometimes require that a single party learns the value of a secret that is shared among all parties. Since all messages are sent over a broadcast channel, we use two layers of secret sharing to obtain the above requirements as described below.
\begin{construction}[secret sharing with respect to a certain party]\label{construct:shareWithRespect}
Let $s$ be a secret taken from some finite field $\FF$.
We share $s$ among $\partNum$ parties  
{\em with respect to} a (special) party $p_j$ in an 
$\alpha$-out-of-$\partNum$ secret-sharing scheme as follows:
\begin{enumerate}
        \item Choose shares $\vect{s^{(1)},s^{(2)}}$ of the secret $s$
                        in a two-out-of-two secret-sharing scheme
                        (that is, select $s^{(1)}\in \FF$ uniformly at random
                        and compute $s^{(2)} = s-s^{(1)}$).
              Denote these shares by $\masking{s}{j}$ and $\compl{s}$, respectively.
              
        \item Compute shares 
        $\vect{\lambda^{(1)},\ldots,\lambda^{(j-1)},\lambda^{(j+1)},\ldots,\lambda^{(\partNum)}}$ 
                                                        of the secret $\compl{s}$ in an $(\alpha-1)$-out-of-$(\partNum-1)$
                                                        Shamir's
secret-sharing scheme.                        
        For each $\ell \neq j$, denote $\complParty{s}{\ell} = \lambda^{(\ell)}$.                                
\end{enumerate}
{\bf Output:}
\begin{itemize}
        \item The share of party $p_j$ is $\masking{s}{j}$. We call this share ``$p_j$'s \emph{masking} share''.
        \item The share of each party $p_{\ell}$, where $\ell\neq j$, is $\complParty{s}{\ell}$.
                                                        We call this share ``$p_\ell$'s \emph{complement} share''.
\end{itemize}
\end{construction}
In the above scheme, we share the secret $s$ among the parties in $P$
in an $\alpha$-out-of-$\partNum$ secret-sharing scheme where only sets of size $\alpha$
that contain $p_j$ can reconstruct the secret. 
In this construction, for
every $\beta<\alpha$ parties, the shares of these parties are uniformly
distributed and independent of the secret.
Furthermore, given such $\beta<\alpha$ shares and a secret $s$, one can {\em efficiently} complete them to $\partNum$ shares of the secret $s$.
In addition, given $\beta$ shares and a secret $s$, one can {\em efficiently} select uniformly at random a
vector of shares competing the $\beta$ shares to $\partNum$ shares of $s$.

\section{Proof of $1/p$-Security of the Protocols with a Dealer}
\label{app:OnlineProof}
In this section we prove that our protocols described in \secref{MPC} that assume an trusted dealer are perfect $1/\poly$-secure implementations of the ideal functionality $\F$. We start by presenting in \appref{simulatorOnLineFromDomain} a simulator for Protocol~$\MPCWithDealer_\numRounds$. In \appref{proofOnLineSimulatorFromDomain}, we prove the correctness of the simulation by showing the the global output in the ideal-world is distributed within $1/\poly$ statistical distance from the global output in the real-world. In \appref{simulatorOnLineFromRange}, we describe the required modifications to the simulator for the protocol for $\F$ that has a polynomial-size range, and argue that the modified simulation is correct.

\subsection{The Simulator for Protocol~$\MPCWithDealer_\numRounds$}
\label{app:simulatorOnLineFromDomain}

We next present a simulator $\simDealer$ for Protocol $\MPCWithDealer_\numRounds$,
described in \figref{MPCWithDealer}.
Let $\badI$ be the set of indices of corrupted parties in the execution. 

The simulator $\simDealer$ invokes $\A$ on the set of inputs $\set{\initialInput{j} : j \in \badI}$, the security parameter $1^\secParam$, and the auxiliary input $\aux$, playing the role of the trusted dealer in the interaction with $\A$.
\begin{description}
\item [Simulating the preprocessing phase:] \quad
\begin{enumerate}
\item   
$\abortedAfterPre = \emptyset$.
\item   
The simulator $\simDealer$ sends a $\startMSG$ message to all corrupt parties.
\item   
$\simDealer$ receives a set of inputs $\set{\partyInput{j} : j \in \badI}$ that $\A$
submits to the computation of the dealer. 
If $\A$ does not submit an input on behalf of $p_j$, i.e., $\A$ sends an $\abort{j}$ message, then, the simulator $\simDealer$ notifies all corrupted parties that party $p_j$ aborted and updates $\abortedAfterPre = \abortedAfterPre \cup \set{j}$.
\item  
$\simDealer$ sets $\abortedI = \abortedAfterPre$. If $\size{\abortedI} \geq \goodNum$, 
the simulator  sets $i=1$ and proceeds to simulate the premature termination step. 
\item   
$\simDealer$ selects $\istar \in \set{1, \ldots, \numRounds}$ with uniform distribution.
\item  
\label{stp:valuesBeforeIstar}  For each $i\in \set{1,\ldots,\istar-1}$ and for each $J \subseteq \badI \setminus \abortedAfterPre$ s.t. $\goodNum \leq \size{J} \leq \badNum$ do

        \begin{enumerate}                                        
        \item \label{stp:valuesBeforeIstarSelection}
        For each $j \in \range{\partNum}$, if $j\in J$, then $\simDealer$ sets $\randomDomain{j}{i} = \partyInput{j}$,
        else, $\simDealer$ selects uniformly at random $\randomDomain{j}{i} \in \partyDomain{}$.
        \item 
        $\simDealer$ sets $\valueMultiParty{i}{J} \gets f_\secParam(\randomDomain{1}{i},\ldots,\randomDomain{\partNum}{i})$.
        \end{enumerate}

\item 
The simulator $\simDealer$ sends $\proceedMSG$ to all corrupt parties.
\end{enumerate}

\item[Simulating interaction rounds:] In each round $1\le i\le \numRounds$, 
   the simulator $\simDealer$ interacts in three phases with the parties $\set{p_j : j\in \badI \setminus \abortedAfterPre}$, i.e., the corrupt parties which are active so far:
\begin{itemize}
  \item {\bf The peeking phase:} 
                
\begin{itemize}

\item If $i=\istar$, the simulator $\simDealer$ sends the set of inputs 
$\set{\partyInput{j} : j \in \badI \setminus \abortedAfterPre}$ to the trusted party computing $\F$
and receives $\outValueSim$.

\item For each  $J \subseteq \badI \setminus \abortedAfterPre$ s.t. $\goodNum \leq \size{J} \leq \badNum$ do
\begin{enumerate}
\item If $i\in \set{1,\ldots,\istar-1}$, the simulator $\simDealer$ sends the value $\valueMultiParty{i}{J}$ (prepared in the simulation of the preprocessing phase) to all parties in $\superSet{J}$ (i.e., to the adversary).

\item Else, if $i\in \set{\istar,\ldots,\numRounds}$, $\simDealer$ sends the value $\outValueSim$ to all parties in $\superSet{J}$ (i.e., to the adversary).

\end{enumerate}                                 

\end{itemize}

\item   {\bf The abort phase:} 
Upon receiving an $\abort{j}$ message from a party $p_j$, 
\begin{enumerate}
\item $\simDealer$ notifies all corrupted parties that party $p_j$ aborted.
\item $\simDealer$ updates $\abortedI = \abortedI \cup \set{j}$.
\item If at least $\goodNum$ parties have aborted so far, that is $\size{\abortedI} > \goodNum$, the simulator $\simDealer$ proceeds to  
simulate the premature termination step.
\end{enumerate}
\item {\bf The main phase:} $\simDealer$ sends $\proceedMSG$ to all corrupt parties.
\end{itemize}

\item[Simulating the premature termination step:]\quad
\begin{itemize}
\item If the premature termination step occurred in round $i=1$,

\begin{itemize}
\item   The simulator $\simDealer$ receives a set of inputs $\set{\partyInput{j}' : j \in \badI\setminus \abortedI}$ that $\A$
submits to the computation of the dealer. \\
If $\A$ does not submit an input on behalf of $p_j$, i.e., sends an $\abort{j}$ message, then, the simulator $\S$
notifies all corrupted parties that party $p_j$ aborted and updates $\abortedI = \abortedI \cup \set{j}$.
\item The simulator $\simDealer$ sends the set of inputs 
$\set{\partyInput{j}' : j \in \badI \setminus \abortedI}$ to the 
dealer and receives $\outValueSim$.
\end{itemize}

\item If the premature termination step occurred in round $1<i<\istar$,
\begin{enumerate}
\item    Upon receiving an $\abort{j}$ message from a party $p_j$, the simulator $\simDealer$ 
 updates $\abortedI = \abortedI \cup \set{j}$.

\item The simulator $\simDealer$ sends the set of inputs $\set{\partyInput{j} : j \in \badI \setminus \abortedI}$ to the trusted party computing $\F$
and receives $\outValueSim$.
\end{enumerate}

\item ($\diamond$ If the premature termination step occurred in round $\istar\leq i \leq \numRounds$, then $\simDealer$ already has $\outValueSim$ $\diamond$)
\item $\simDealer$ sends the value $\outValueSim$ to each party  
in $\set{p_j : j\in \badI \setminus \abortedAfterPre}$.        

\end{itemize}

\noindent  
\item[Simulating normal termination:]
If the last round of the protocol is completed, then $\simDealer$ sends $\outValueSim$ to each party  
in $\set{p_j : j\in \badI \setminus \abortedAfterPre}$.        
\end{description}
At the end of the interaction with $\A$, the simulator will output 
the sequence of messages exchanged between the simulator and the corrupted parties.

\subsection{Proof of the Correctness of the Simulation for $\MPCWithDealer_\numRounds$} 
\label{app:proofOnLineSimulatorFromDomain}
In order to prove the correctness of the simulation described in \appref{simulatorOnLineFromDomain},
we consider the two random variables from \secref{realIdeal}, both of the form 
$(\viewVAR,\outValueVAR)$, where $\viewVAR$ describes a possible view of $\A$, and 
$\outValueVAR$ describes a possible output of the honest parties (i.e., 
$\outValueVAR\in \partyRange{}$). The first random variable 
$\Real_{\MPCWithDealer_{\numRounds},\A(\aux)}(\vecInput,1^\secParam)$ 
describes the real world -- an execution of Protocol $\MPCWithDealer$, where $\viewVAR$ describes the view of the adversary $\A$ in this execution, and 
$\outValueVAR$ is the output of the honest parties in this execution. 
The second random variable $\Ideal_{\F,\simDealer(\aux)}(\vecInput,1^\secParam)$ describes the 
ideal world -- an execution with the trusted party computing $\F$ (this trusted party is denoted by $T_{\F}$), where $\viewVAR$ 
describes the output of the simulator $\simDealer$ in this execution, and $\outValueVAR$ is the 
output of the honest parties in this execution.
For the rest of this section, we simplify
notations and
denote the above two random variables by  $\Real=
\vect{\viewVAR_{\Real},\outValueVAR_{\Real}}$ and  $\Ideal=
\vect{\viewVAR_{\Ideal},\outValueVAR_{\Ideal}}$ respectively.

We consider the probability of a given pair $\vect{\viewVar,\outValueVar}$ according 
to the two different random variables. 
We compare the two following probabilities: (1) The probability that $\viewVar$ is the view of the adversary $\A$ in an    
   execution of Protocol $\MPCWithDealer_\numRounds$ and $\outValueVar$ is the output 
   of the honest parties in this execution,
   where the probability is taken over the random coins of the dealer $T$.
 (2) The probability that $\viewVar$ is the output of the 
  simulator $\simDealer$ in an ideal-world execution with the trusted party $T_{\F}$ and 
  $\outValueVar$ is the output of the honest parties in this execution, 
   where the probability is taken over the random coins of the simulator $\simDealer$ and the 
   random coins of the ideal-world trusted party $T_{\F}$.

In \lemref{simulatorD} we prove the correctness of the simulation by showing that the two random variables are within statistical distance $1/\poly$.
For the proof of the lemma we need the following claim from ~\cite{GK10}.
\begin{claim}[{\cite[Lemma 2]{GK10}}]\label{clm:alphaGuessingi*}
Let $\A$ be an adversary in Protocol $\MPCWithDealer_\numRounds$ and 
let $\partyInput{1},\ldots,\partyInput{\partNum}$ be a set of inputs.
Assume that for every possible output $\outValue$ obtained by the dealer using this set of inputs
the probability that in a round $i<\istar$ all the values that the adversary sees are equal to $\outValue$
is at least $\alpha$.
Then, the probability that $\A$ guesses $\istar$ (i.e., causes premature termination in round $\istar$) is 
at most $1/\alpha\numRounds$.
\end{claim}

 As the adversary might have some auxiliary information
 on the inputs of the honest parties and know the value of
 $f_\secParam(x_1,\dots,x_\partNum)$, the adversary might be able to deduce
 that a round is not $\istar$ if not all the values that it gets are equal
 to this value (or a possible value for randomized functionalities).
 Specifically, in the worst case scenario, the adversary knows the inputs of all the honest parties. In the next
 claim we show a lower bound on the probability that all the values that
 the adversary obtains in a round $i < \istar$ of Protocol $\MPCWithDealer_\numRounds$ are all equal to a fixed
 value.

\begin{claim}
\label{clm:equalValuesFromDomain}
Let $\domainSize(\secParam)$ and $\rangeSize(\secParam)$ be the size of the domain and range, respectively, of a randomized functionality $\F$ computed by the protocol $\MPCWithDealer_\numRounds$.
Let $\epsilon$ be a number such that 
$\Pr [f_\secParam(\partyInput{1},\ldots,\partyInput{\partNum}) = \outValue_\ell] \geq \epsilon$
for every set of inputs $\partyInput{1},\ldots,\partyInput{\partNum}$ 
and for each $\outValue_\ell$ from the range of $f_\secParam(\partyInput{1},\ldots,\partyInput{\partNum})$.
Then, 
the probability that in a round $i<\istar$ all the values that the adversary sees are equal to a specific $\outValue$
                is at least $\left(\epsilon/\domainSize(\secParam)^{\partNum} \right)^{2^\badNum-1}$.            
                
                Furthermore,  if $\F$ is deterministic, then, this probability is at least $(1/\domainSize(\secParam)^{\partNum})^{2^\badNum-1}$.
\end{claim}
\begin{proof}
We start with the case of a deterministic functionality $\F$.
Recall that $\partyInput{1},\ldots,\partyInput{\partNum}$ are the inputs used by the dealer to obtain $\outValue = f_\secParam(\partyInput{1},\ldots,\partyInput{\partNum})$ 
and $\valueMultiParty{i^\star}{J} = \outValue$ for each $J \subseteq \range{\partNum}$ s.t. $\goodNum \leq \size{J} \leq \badNum$.
Let $J$ be such that the adversary obtains $\valueMultiParty{i}{J}$ in round $i<\istar$.
Recall that $\randomDomain{1}{i},\ldots,\randomDomain{\partNum}{i}$ are the inputs used by the dealer to obtain $\valueMultiParty{i}{J}$,
that is, $\valueMultiParty{i}{J} = f_\secParam(\randomDomain{1}{},\ldots,\randomDomain{\partNum}{})$,
where $\randomDomain{j}{i}=\partyInput{j}$ for each $j\in J$
and $\randomDomain{j}{i}$ is selected uniformly at random from $\randomDomain{j}{}$ for every $j\notin J$.
We bound the probability that $\valueMultiParty{i}{J}=\outValue$ by the probability that $\randomDomain{j}{i}=\partyInput{j}$ for all $j\notin J$.
The probability that $\randomDomain{j}{i}=\partyInput{j}$ is $1/\domainSize$.
Therefore, the probability that both sets are the same is $(1/\domainSize)^{\partNum - \size{J}} > (1/d)^{\partNum}$.

In each round of the protocol, $\A$ obtains the value
$\valueMultiParty{i}{J}$  for each subset $\superSet{J}$ s.t. $J \subseteq
\range{\partNum}$ and $\goodNum \leq \size{J} \leq \badNum$, therefore,
$\A$ obtains less than  $2^\badNum$ values.
For each such two values $\valueMultiParty{i}{J}$ and $\valueMultiParty{i}{J'}$ obtained by $\A$ in round $i<\istar$, 
the sets of inputs $\set{\randomDomain{j}{i} : j\notin J}$ and $\set{\randomDomain{j}{i} : j\notin J'}$
are totally independent.
Therefore, the probability that all the values that the adversary sees in round $i<\istar$ are equal to 
$\outValue = f_\secParam(\partyInput{1},\ldots,\partyInput{\partNum})$ 
is at least $(1/\domainSize^{\partNum})^{2^\badNum-1}$.

For randomized functionality $\F$, 
we think of the evaluation of $f_\secParam(\randomDomain{1}{},\ldots,\randomDomain{\partNum}{})$ as two steps:
first $\randomDomain{j}{}$ is randomly chosen from $\partyDomain$ for every $j \not\in J$ and
then the randomized functionality is evaluated. 
Therefore, as $\A$ obtains less than $2^\badNum$ values in each round $i<i^\star$,
that the probability that all the values that the adversary sees in each round $i<i^\star$
are equal to the specific $\outValue$ is at least $(1/\domainSize^{\partNum})^{2^\badNum-1} \cdot \epsilon^{2^\badNum-1}$.
\end{proof}

\medskip
In the next lemma, we prove the correctness of the simulation by using the previous two lemmas.
\begin{lemma}
\label{lem:simulatorD}
Let $\F$ be a (possibly randomized) functionality, $\A$ be a non-uniform polynomial-time adversary corrupting $\badNum<2\partNum/3$ parties in an execution of Protocol $\MPCWithDealer$, and $\simDealer$ be 
the simulator described in \appref{simulatorOnLineFromDomain} (where $\simDealer$ controls the same parties as $\A$). Then, for every $\secParam\in\NN$, for every $\vecInput \in (\partyDomain)^\partNum$, and for every ${\rm aux} \in\set{0,1}^*$
 $$\operatorname{SD}{\Big(\Real_{\MPCWithDealer_{\numRounds},\A(\aux)}(\vecInput,1^\secParam),} 
{\Ideal_{\F,\simDealer(\aux)}(\vecInput,1^\secParam)}\Big) 
 \leq 2 \rangeSize(\secParam)  {\domainSize(\secParam)^\partNum} / \left(\numRounds(\secParam)\right)^{2^\badNum},$$
where $\domainSize(\secParam)$ and $\rangeSize(\secParam)$ are the sizes of the range and the domain of $\F$, respectively, and $\numRounds(\secParam)$ be the number of rounds in the protocol.

Furthermore, if $\F$ is deterministic, then, the statistical distance between these two random variables is at most
$({\domainSize(\secParam)^\partNum  })^{2^\badNum}/\numRounds(\secParam)$.
\end{lemma}
\begin{proof}
Our goal here is to show that the  statistical distance between the above two
random variables is at most as described in lemma.
The flow of our proof is as follows. We first bound the statistical distance between the two random variables by the probability that the
adversary $\A$ guesses the special  round $\istar$. We do this by showing
that, conditioned on the event that the adversary  fails to guess round
$\istar$, the two random variables are identically distributed. Then,  we
bound the probability of guessing $\istar$ in time using \clmref{alphaGuessingi*} and \clmref{equalValuesFromDomain}. 

Observe that, in the simulation, $\simDealer$ follows the same instructions as the
trusted party  $T$ in Protocol $\MPCWithDealer_\numRounds$, except for two
changes.  First, $\simDealer$ does not compute the output $\outValueSim$,  but rather gets
$\outValueSim$ externally from $T_{\F}$.  The simulator obtains this value
either in the premature termination phase (if $i<\istar$) or in the
peeking stage when $i=\istar$.  The second difference is that in the case
of a premature termination, $\simDealer$  will always use $\outValueSim$ as its
message to the corrupt parties, while $T$ will use  the value from round
$\istar-1$ of the appropriate subset $Q_J$ as its message.
%

 We analyze the probabilities of $\vect{\viewVar,\outValueVar}$ in the two random variables according to weather the  
 premature termination occurred before, during, or after the special round $\istar$.
\paragraph {Premature termination before round 
{\boldmath $\istar$}.}  We argue that in this case, both in the real
protocol and in the simulation,  the view of $\A$ is identically
distributed in the two worlds.  $\simDealer$ follows the same random process in
interacting with $\A$ (before sending the last message in the premature
termination) as does $T$ in the real-world  execution.  The view of the adversary consists
of values which are outputs of evaluations of the function $f_\secParam$ on
the same input distributions. 
The adversary does not learn anything about the inputs of the honest parties, hence,  
its decision to abort does not depend on any new information it obtains during the interaction rounds so far.
In addition, in both
worlds, the output of the honest parties is the evaluation of the function
$f_\secParam$ on the same set of inputs for the active parties and
uniformly selected random inputs for the aborted parties.
 
\paragraph {Premature termination after round 
\boldmath{$\istar$} or never occurs.} 
Here $\viewVar$ must contain $\valueMultiParty{\istar}{J}$ for some $J$, which, in the 
real-world execution, is equal to the output value of all sets for any round 
$i > \istar$ (recall that the output value of the honest parties will be determined by 
one such value), and in the simulation it equals $\outValueSim$. Thus, in both 
scenarios, $\viewVar$ must be consistent with $\istar$ and with $\outValueVar$, hence,
$\viewVar$ completely determines $\outValueVAR$. Again, since $\simDealer$ follows the same 
random process in interacting with $\A$ as does $T$ in the real-world execution
the probabilities are the same. 

 \paragraph {Premature termination in round \boldmath{$\istar$}.}
 This is the interesting case, which causes the statistical distance.
 In the real world, the output of the honest parties is
 $\valueMultiParty{\istar-1}{J}$ for some $J$, while in the ideal world their output is
 $\outValueSim \gets f_\secParam(\partyInput{1},\ldots,\partyInput{\partNum})$.
 In the first case the output is independent of the adversary's view,
 while in the second case, the view determines the output.
 Thus, in this case the probabilities of the views are different. However, we will
 show that the event of premature termination in round $\istar$ happens with small probability.

 \medskip

 Since the probabilities of $\vect{\viewVar,\outValueVar}$ in the first two cases are equal, the statistical
 distance between the two random variables is bounded by the probability of
 the adversary  guessing $\istar$ correctly (before the abort phase of
 round $\istar$).  
 That is,
 \begin{eqnarray}\label{eqn:SD}
 \SD{\Ideal}{\Real} \le\pr[\text{Premature termination in round } \istar].
 \end{eqnarray}

We next use \clmref{alphaGuessingi*} and \clmref{equalValuesFromDomain} to bound the probability
that the adversary guesses $i^\star$. However, there might be values such that 
$\pr[\outValue = f_\secParam(\partyInput{1},\ldots,\partyInput{\partNum})]$
is small. Therefore, we consider two events of guessing $i^\star$,
where $p_0$ is a parameter specified below.
We call an output values $\outValue$ \emph{heavy} if
$\pr[\outValue = f_\secParam(\partyInput{1}{},\ldots,\partyInput{\partNum}{})] > 1 / (p_0 \cdot \rangeSize)$, otherwise, we call $\outValue$ \emph{light}.
\begin{description}
        \item[Case 1:] The adversary guesses $i^\star$ with some light $\outValue$.
                                                                        Since there are at most $\rangeSize$ possible values of 
                                                                        $f_\secParam(\partyInput{1}{},\ldots,\partyInput{\partNum}{})$,
                                                                        the probability of this event, by the union bound, is at most $1 / p_0$.
        \item[Case 2:]  The adversary guesses $i^\star$ with some heavy $\outValue$.
                        Thus, by \clmref{equalValuesFromDomain} where $\epsilon = p_0 \cdot \rangeSize$, the probability of
                        $\outValue =\valueMultiParty{i}{J}$ for all values that the adversary
                                                                        sees in round $i<i^\star$ is at least 
                                                                        $(1/\domainSize^{\partNum}\cdot p_0\cdot \rangeSize)^{2^\badNum-1}$.
                                                                        By \clmref{alphaGuessingi*}, the probability that the adversary guesses
                                                                        $i^\star$ conditioned on the $\outValue$ being heavy is at most
                                                                        $(\domainSize^{\partNum} \cdot p_0 \cdot \rangeSize)^{2^\badNum-1} / \numRounds$.
\end{description}
%

We take
$p_0 = \numRounds^{2^{-\badNum}} / (\rangeSize\cdot \domainSize^{\partNum})$;
the total probability that the adversary guesses $i^\star$ in the two cases is at most
$$\frac{(\domainSize^{\partNum} \cdot p_0 \cdot \rangeSize)^{2^\badNum-1}}{\numRounds} + \frac{1}{p_0} \leq 2 \cdot \frac{\rangeSize\cdot \domainSize^{\partNum}}{\numRounds^{2^{-\badNum}}} . $$
Therefore, by \eqnref{SD}, the
statistical distance between the two random variables in the randomized case is as claimed in the
lemma.

The case that $\F$ is deterministic is simpler. By combining
\clmref{alphaGuessingi*} and \clmref{equalValuesFromDomain} we get that the
probability that $\A$ guesses $i^\star$ is at most
$(\numRounds/\domainSize(\secParam)^{\partNum})^{2^\badNum-1}$.  By
applying \eqnref{SD}, we get the bound on statistical distance between the
two random variables for the deterministic case as claimed in the lemma.
%
%
\end{proof}

\subsection{The Simulator for the Protocol with the Dealer for Polynomial Range}
\label{app:simulatorOnLineFromRange}
\begin{lemma}
Let $\F$ be a (possibly randomized) functionality.
For every non-uniform polynomial-time adversary $\A$ corrupting $\badNum<2\partNum/3$ parties in an execution of Protocol $\MPCWithDealerRange$, there exists a simulator $\simDealer$ in the ideal model,
that simulates the execution of $\A$ (where $\simDealer$ controls the same parties as $\A$).
That is, for every $\secParam\in\NN$, for every $\vecInput \in (\partyDomain)^\partNum$, and for every ${\rm aux} \in\set{0,1}^*$ 
$$\operatorname{SD}\left(\Real_{\MPCWithDealer_{\numRounds},\A(\aux)}(\vecInput,1^\secParam), 
{\Ideal_{\F,\simDealer(\aux)}(\vecInput,1^\secParam)}\right) < 
\frac{\left(2p(\secParam) \cdot \rangeSize(\secParam)\right)^{2^\badNum}}{\numRounds(\secParam)}
 + \frac{1}{2p(\secParam)},$$
where $\rangeSize(\secParam)$ is the size of the range of $\F$, with probability $1/(2p(\secParam))$ each value $\valueMultiParty{i}{J}$ in round $i <i^\star$ is selected uniformly at random from the range, and $\numRounds(\secParam)$ be the number of rounds in the protocol.
%
\end{lemma}

\begin{proof}
The simulators and their proofs for Protocol $\MPCWithDealerRange$ and
Protocol $\MPCWithDealer$ are similar;
we only present (informally) the differences between the two simulators and the two proofs.

\paragraph{The modified simulator.}
Recall that the protocols $\MPCWithDealer$ and $\MPCWithDealerRange$ are different only in
\stepref{beforeIStar} of the share generation step.
In $\MPCWithDealerRange$, each value $\valueMultiParty{i}{J}$ prior to round $\istar$
is chosen with probability $1/(2p)$ as a random value
from the range of $f_\secParam$ and with probability $1-1/(2p)$ it is chosen
just like in \figref{MPCWithDealer}.
There are two modifications to the simulator.
The first modification in the simulator is in \stepref{valuesBeforeIstar} in the simulation of the preprocessing phase,
i.e., in the computation of $\valueMultiParty{i}{J}$ for $i<\istar$.
The step that replaces \stepref{valuesBeforeIstar} appears below.
\begin{itemize}
                              \item                   
For each $i\in \set{1,\ldots,\istar-1}$ and for each $J \subseteq \badI \setminus \abortedAfterPre$ s.t. $\goodNum \leq \size{J} \leq \badNum$ do 
                        
                                \begin{enumerate}
                                \item with probability $1/(2p)$, select uniformly at random $z_J^i \in \partyRange{}$ and set $\valueMultiParty{i}{J} = z_J^i$.
                                                                        \item   with the remaining probability $1-1/(2p)$,
                                                                                \begin{enumerate}
                                                                                        \item For each $j \in \range{\partNum}$, if $j\in J$, then $\simDealer$ sets $\randomDomain{j}{i} = \partyInput{j}$,
                                 else, $\simDealer$ selects uniformly at random $\randomDomain{j}{i} \in \partyDomain{}$.
                                 \item $\simDealer$ sets $\valueMultiParty{i}{J} \gets f_\secParam(\randomDomain{1}{i},\ldots,\randomDomain{\partNum}{i})$.
                                                                                \end{enumerate}

\end{enumerate}                                 
                                                                                                 
                                                                \end{itemize}    
The second modification is less obvious.
Recall that both random variables appearing in the lemma contain the output of the honest parties.
In the ideal world, the honest parties always output $f_\secParam$ applied to their inputs.
In the real world, in a premature termination in round $i<\istar$, with probability $1/(2p)$, the honest parties output 
a random value from the range of $f_\secParam$. It is hard to simulate the output of the honest parties in first case.\footnote{For example, there might not be possible inputs of the corrupt parties causing the honest parties to output such output.} We simply
modify the simulator such that with probability $1/(2p)$ the simulator returns $\bot$, i.e., it announces that the simulation has failed.
The new premature termination step appears below.
\begin{description}
\item[Simulating the premature termination step:] \quad 
\begin{itemize}
        \item If the premature termination step occurred in round $i<\istar$,
                \begin{itemize}
                        \item With probability $1/(2p)$, for each $j\in \badI \setminus \abortedAfterPre$ send $\abort{j}$ to the trusted party computing $\F$ and return $\bot$.
                        \item With the remaining probability $1-1/(2p)$, execute the original simulation of the premature termination step 
                                        (appearing in \appref{simulatorOnLineFromDomain}).
                \end{itemize}
        \item Else ($i\geq i^\star$), execute the original simulation of the premature termination step (appearing in \appref{simulatorOnLineFromDomain}).
\end{itemize}
\end{description}

\paragraph{The modified proof.}
The proof to the simulator for $\MPCWithDealerRange$ remains basically the same, except for two changes.
We first modify \clmref{equalValuesFromDomain} below and prove a slightly different claim, which changes the probability of the adversary guessing $i^\star$.

\begin{claim}
\label{clm:equalValuesFromRange}
                Let $\rangeSize(\secParam)$ be the size of the range of the (possibly randomized) functionality $\F$
                computed by the protocol $\MPCWithDealerRange_\numRounds$ and $\outValue\in \partyRange$.
                Then, the probability that in a round $i<\istar$ all the values that the adversary sees are 
                equal to $\outValue$ is at least 
                $(1/2p(\secParam) \cdot \rangeSize(\secParam))^{2^\badNum}$.
\end{claim}
\begin{proof}
According to the protocol, there are two different ways to produce each
value $\valueMultiParty{i}{J}$ in round $i<\istar$: (1) Compute
$f_\secParam$ on a set of inputs and a set of uniformly selected values
from the domain of the functionality, and  (2) Set $\valueMultiParty{i}{J}$
as a uniformly selected value from the range of the functionality.
We ignore the first case. In the second option, with probability
$1/2p$, the value $\valueMultiParty{i}{J}$ is uniformly selected from
the range.  Hence, the probability that $\valueMultiParty{i}{J}$ is equal
to a specific value is at least $1/(2p \cdot \rangeSize)$.

It was explained in the proof of \clmref{equalValuesFromDomain} that in each round of the protocol, $\A$ obtains less than  $2^\badNum$ values.
Therefore, we conclude that he probability that all the values that $\A$ obtains in round $i<\istar$ are all equal
to $\outValue$ is at least $(1/(2p \cdot \rangeSize))^{2^\badNum}$.
\end{proof}

\medskip By applying the \clmref{alphaGuessingi*} we conclude that the
probability of the adversary guessing $\istar$ correctly in Protocol
$\MPCWithDealerRange_\numRounds$ is at most $(2p \cdot
\rangeSize)^{2^\badNum}/\numRounds$.  In case of a premature termination in
round $i<i^\star$,  with probability $1-1/(2p)$ in both the ideal world and
real world, the value that the honest parties  output is the evaluation of
$f_\secParam$ on the inputs of the active parties and random inputs for the
parties that aborted. However, with probability  $1/(2p)$, if premature
termination occurs prior to round $i^\star$, the output of the honest
parties Protocol $\MPCWithDealerRange_\numRounds$ is a random value from
the range of  $f_\secParam$; the simulator fails to simulate the execution
in this case and outputs $\bot$. Thus,
 \begin{eqnarray*}
\lefteqn{ \SD{\Ideal}{\Real} }\\
 & \le &\pr[\text{Premature termination in round } \istar]
 +(1/2p)\cdot \pr[\text{Premature termination before round } \istar]\\
 & \le & (2p \cdot \rangeSize)^{2^\badNum}/\numRounds +(1/2p).
\end{eqnarray*}
Therefore, the statistical distance is
as claimed.
\end{proof}

\section{Proof of Security for the Protocols without the Dealer}

\subsection{The Simulator for Protocol $\MPC_\numRounds$}
\label{app:OfflineProof}

We next prove that Protocol $\MPC_\numRounds$ is a secure real-world
implementation of the (ideal) functionality of Protocol $\MPCWithDealer_\numRounds$.
By \lemref{simulatorD}, when $\numRounds(\secParam)$ is sufficiently large, Protocol $\MPCWithDealer_\numRounds$
is a $1/p$-secure protocol for $\F$. Thus, together we get that Protocol $\MPC_\numRounds$ is a 
$1/p$-secure protocol for $\F$.
according to the definition appears in \appref{securityWithBlaming}.
We analyze Protocol $\MPC_\numRounds$ in a hybrid model where there are 3 ideal functionalities:
\begin{description}

\item [Functionality {\boldmath $\MultiShareGenDomainWithAbort_\numRounds$}.]
This functionality is an (ideal) execution of Functionality 
$\MultiShareGenDomain_\numRounds$ in the secure-with-abort and cheat-detection model.
That is, the functionality gets a set of inputs.
If the adversary sends $\abort{j}$ for some corrupt party $p_j$, then this message is sent to the honest parties and the execution terminates.
Otherwise, Functionality $\MultiShareGenDomain_\numRounds$ is executed. 
Then, the adversary gets the outputs of the corrupt parties. Next, the adversary decides whether to halt or to continue:
If the adversary decides to continue, it sends a $\proceedMSG$ message and the honest parties 
are given their outputs.
Otherwise, the adversary sends $\abort{j}$ for some corrupt party $p_j$, and this 
message is sent to the honest parties.
\item [Functionality {\boldmath $\FairMPC$}.]
This functionality computes the value $f_\secParam(x_1,\ldots,x_\partNum)$.
That is, the functionality gets a set of inputs.
If a party $p_j$ sends $\abort{j}$ message then $\partyInput{j}$ selected from $\partyDomain{}$ with uniform distribution,
computes an output of the randomized functionality $f_\secParam$ for them, and gives it to all parties.
When this functionality is executed, an honest majority is 
guaranteed, hence, the functionality can be implemented with full security (e.g., with 
fairness).
\item [Functionality {\boldmath $\Reconstruction$}.] This functionality is described 
in \figref{Reconstruction}; this functionality is used in the premature termination step in Protocol 
$\MPC_\numRounds$ for reconstructing the output value from the shares of the 
previous round. When this functionality is executed, an honest majority is guaranteed, 
hence, the functionality can be implemented with full security (e.g., with fairness).    
\end{description}

We consider an adversary $\A$ in the hybrid model described above, corrupting $\badNum<2\partNum/3$ of the parties that engage in Protocol $\MPC_\numRounds$.
We next describe a simulator $\S$ interacting with the honest parties in the 
ideal-world via a trusted party $T_{\MPCWithDealer}$ executing Functionality $\MPCWithDealer_\numRounds$. The simulator $\S$ runs the adversary $\A$ internally
with black-box access. Simulating $\A$ in an execution of the protocol, $\S$ corrupts 
the same subset of parties as does $\A$. Denote by $\badI=\set{i_1,\ldots,i_{\badNum}}$ the set of indices of corrupt party.
At the end of the computation it outputs a possible view of the adversary $\A$.                 
To start the simulation, $\S$ invokes $\A$ 
on the set of inputs $\set{\initialInput{j} : j \in \badI}$,
the security parameter $1^\secParam$, and the auxiliary input $\aux$.

                        
\begin{description}
\item[Simulating the preliminary phase:]\quad
\label{stp:playRoleOfShareGenWithAbort} 
       
\begin{enumerate}
\item \label{stp:initD} $\abortedAfterPre  = \emptyset$.
\item  \label{stp:getInputs}    The simulator $\S$ receives a set of inputs $\set{\partyInput{j} : j \in \badI\setminus \abortedAfterPre }$ that $\A$
        submits to Functionality $\MultiShareGenDomainWithAbort_\numRounds$. \\
        If a party $p_j$ for $j \in \badI\setminus \abortedAfterPre$ does not submit an input, i.e., sends an $\abort{j}$ message, then,
                                                \begin{enumerate}
                                                                \item $\S$ sends $\abort{j}$ to the trusted party $T_{\MPCWithDealer}$.
                                                                \item $\S$ updates $\abortedAfterPre  = \abortedAfterPre  \cup \set{j}$.
                                                                \item \label{stp:afterAbortingHonestBefore} If $\size{\abortedAfterPre } < \goodNum$, then \stepref{getInputs} is repeated.
                                                                \item \label{stp:afterAbortNoHonestBefore} Otherwise ($\size{\abortedAfterPre } \geq \goodNum$), simulate premature termination with $i=1$.     
                                                \end{enumerate}
 
\item \label{stp:prepareMSGSForAdversary} 
$\S$ prepares outputs for the corrupted parties for Functionality $\MultiShareGenDomainWithAbort_\numRounds$:
The simulator $\S$ sets $\valueMultiParty{i}{J}=0$ for every $J \subseteq \range{\partNum}\setminus \abortedAfterPre $ s.t. $\goodNum \leq \size{J} \leq \badNum$
and for all $i\in \set{1,\ldots,\numRounds}$.
Then, $\S$ follows \stepref{defaultInput} and \steprefs{produceKeys1}{preperMSG} in the 
computation of Functionality $\MultiShareGenDomain_\numRounds$ (skipping the \steprefs{selectIStar}{chooseValues2})
to obtain shares for the parties.\footnote{These shares are temporary and will later be open to the actual 
values obtained from $T_{\MPCWithDealer}$ during the interaction rounds using the 
properties of Shamir's secret-sharing scheme.}

\item \label{stp:sendMessagesToAdversary} 
For each party $p_j$ s.t. $j\in \badI \setminus \abortedAfterPre $, the simulator $\S$ sends to $\A$:
\begin{itemize}
\item The verification key $K_{\rm ver}$.
\item The masking shares $\masking{\valueInnerSecretSharingShareSigned{i}{J}{j}}{j}$
for each $i\in \set{1,\ldots,\numRounds}$ and 
 for every $J \subseteq \range{\partNum}\setminus \abortedAfterPre $ s.t. $\goodNum \leq \size{J} \leq \badNum$ and
  $j\in J$.
\item The messages $M_{j,1}, \ldots, M_{j,\numRounds}$.

\end{itemize}
\item \label{stp:abortAfterMessages}If $\A$ sends an $\abort{j}$ for some party $p_j$ s.t. $j\in \badI \setminus \abortedAfterPre$ to $\S$,
then,
\begin{enumerate}
\item $\S$ sends $\abort{j}$ to the trusted party $T_{\MPCWithDealer}$.
\item $\S$ updates $\abortedAfterPre  = \abortedAfterPre  \cup \set{j}$.
\item \label{stp:afterAbortingHonest} If $\size{\abortedAfterPre } < \goodNum$, then
\steprefs{getInputs}{abortAfterMessages} are repeated.
\item \label{stp:afterAbortNoHonest} Otherwise ($\size{\abortedAfterPre } \geq \goodNum$), go to simulating premature termination with $i=1$.
\end{enumerate}

                                Otherwise ($\A$ sends a $\continueMSG$ message to $\S$), 
                                
\begin{enumerate}
        
        \item The simulator $\S$ denotes $\abortedI = \abortedAfterPre$.
        
   \item The simulator sends $\partyInput{j}$ to $T_{\MPCWithDealer}$ for every $j\in \badI \setminus \abortedAfterPre$
                                (and gets as response a $\proceedMSG$ message).
\end{enumerate}
                                
\end{enumerate} 

\item[Simulating interaction rounds:]\quad\\ 
Let $\J$ be the collection of subsets $J \subseteq \badI \setminus \abortedAfterPre$ s.t. $\goodNum \leq \size{J} \leq \badNum$.
I.e., $\J$ is the collection of sets of indices of active corrupt parties after the simulation of the executions of $\MultiShareGenDomainWithAbort_\numRounds$
To simulate round $i$ for $i=1,\ldots,\numRounds$, the simulator $\S$ proceeds as follows:
\begin{enumerate}  
\item \label{stp:getBitsInRound} $\S$ gets from the trusted party  $T_{\MPCWithDealer}$ 
 the values that the corrupted parties see. 
 That is, $\S$ gets a bit $\valueMultiPartyFromT{i}{J}$ for each $J \in \J$.\footnote{In \steprefs{innerSharesConstruction}{constructionSignedMSGS}, 
 the simulator $\S$ constructs the messages of the honest parties in order to 
 allow the corrupted parties in each $J \in \J$ to reconstruct $\valueMultiPartyFromT{i}{J}$.}

\item \label{stp:innerSharesConstruction}  
The simulator $\S$ selects shares for the inner secret-sharing scheme for corrupted parties:
For every $J \in \J$, the simulator $\S$ selects uniformly at random shares of $\valueMultiPartyFromT{i}{J}$ in a $\size{J}$-out-of-$\size{J}$ Shamir secret sharing scheme.
Denote these shares by $\set{\valueInnerSecretSharingShareFromT{i}{J}{j} : p_j \in \superSet{J}}$.\\ 
For each $p_j \in \superSet{J}$, let
$\valueInnerSecretSharingShareSignedFromT{i}{J}{j}\leftarrow
   (\valueInnerSecretSharingShareFromT{i}{J}{j}, i, J, j ,
   \operatorname{Sign}((\valueInnerSecretSharingShareFromT{i}{J}{j}, i, J, j) ,K_{\rm sign})).$

\item \label{stp:outerSharesConstruction}
The simulator $\S$ selects complementary shares for all honest parties:
For every $J \in \J$ and for each $j \in \badI \setminus \abortedAfterPre$,
\begin{enumerate}
        \item $S$ calculates $\alpha_j = \masking{\valueInnerSecretSharingShareSigned{i}{J}{j}}{j} \xor \valueInnerSecretSharingShareSignedFromT{i}{J}{j}$.
        \item $S$ selects uniformly at random $\goodNum$ shares of $\alpha_j$ uniformly at 
                                random over all possible selections of $\goodNum$ shares that 
                                are shares of $\alpha_j$ together with
                                the $\size{\badI \setminus \abortedAfterPre}-1$ shares 
             $$\set{\complParty{\valueInnerSecretSharingShareSigned{i}{J}{j}}{q}: q\in \badI \setminus (\abortedAfterPre \cup \set{j})}$$
                                produced in \stepref{prepareMSGSForAdversary} in the simulation of the preliminary phase.\\
                                (This is possible according to the property of Shamir's scheme)\\
                                Denote by $\complParty{\valueInnerSecretSharingShareSignedFromT{i}{J}{j}}{q}$ 
                                the complementary share that $\S$ selects for the honest party $p_q$ for 
                                a party $p_j$ s.t. $j\in (\badI \setminus \abortedAfterPre )\cap J$, where $J \in \J$.\\
\end{enumerate}
\item \label{stp:complementarySharesFromPreliminaryPhase} 
For party $p_j$ and a subset $J \notin \J$, let $\complParty{\valueInnerSecretSharingShareSigned{i}{J}{j}}{q}$ be the complementary share which was produced in \stepref{prepareMSGSForAdversary} in the simulation of the preliminary phase, i.e., $\complParty{\valueInnerSecretSharingShareSigned{i}{J}{j}}{q}$.

\item \label{stp:constructionSignedMSGS} 
Construct signed messages $m'_{q,i}$ for each honest party $p_q$ in round $i$ by concatenating:
\begin{enumerate}
        \item $q$.
        \item The round number $i$.
        \item The complement shares which were described in \stepref{complementarySharesFromPreliminaryPhase} above.
        \item The complement shares $\complParty{\valueInnerSecretSharingShareSignedFromT{i}{J}{j}}{q}$ 
                                for all $J \in \J$ and for all $j  \in J$  produced in
                                \stepref{outerSharesConstruction} for $p_q$.
\end{enumerate}
Then, $\S$ signs $m'_{q,i}$, i.e.,
$\S$ computes $M'_{q,i} \leftarrow \vect{m'_{q,i},\operatorname{Sign}(m'_{q,i},K_{\rm sign})}$.  
 
\item The simulator $\S$ sends all the message $M'_{q,i}$ on behalf of each honest party
$p_q$ to $\A$.

\item For every $j \in \badI \setminus \abortedAfterPre$ s.t. $\A$ sends an invalid or no message on behalf of $p_j$,
the simulator $\S$ sends 
$\abort{j}$ to $T_{\MPCWithDealer}$:
\begin{enumerate}

\item $\abortedI  = \abortedI  \cup \set{j}$.      

\item If $\size{\abortedI } \geq \partNum - \badNum$ go to premature termination step.                                                                                                                                
\item Otherwise, the simulator $\S$ proceeds to the next round.
\end{enumerate}
\end{enumerate}

\item[Simulating the premature termination step:] \quad
\begin{itemize}
\item \label{stp:prematureProcessFirstStep}
If $i=1$, then $S$ simulates $\A$'s interaction with Functionality $\FairMPC$ as follows:
\begin{enumerate}
\item $S$ receives from $\A$ the inputs of the active corrupt parties.
\item For every $j \in \badI \setminus \abortedI$: If $p_j$ does not send an input, 
 then $\S$ sends $\abort{j}$ to $T_{\MPCWithDealer}$ else, $\S$ sends $p_j$'s input to $T_{\MPCWithDealer}$.
\end{enumerate}   

\item \label{stp:prematureProcessNotFirstStep} If $i>1$, then $S$ simulates $\A$'s interaction with Functionality $\Reconstruction$ as follows:
\begin{enumerate}
\item\label{stp:prematureProcessNotFirstStep1} $S$ receives from $\A$ the inputs of 
 the active corrupt parties, i.e., $p_j$ s.t. $j \in \badI \setminus \abortedI$.
\item If an active corrupt party $p_j$, does not send an input, or its input is not appropriately signed or malformed, 
 then $\S$ sends $\abort{j}$ to $T_{\MPCWithDealer}$.
\end{enumerate}

\item $\S$ gets from $T_{\MPCWithDealer}$ a value $\sigma$ and sends it to $\A$.
 \item The simulator $\S$ outputs the sequence of messages exchanged between 
        $\S$ and the adversary $\A$ and halts.                                                
\end{itemize} 
\item[Simulating normal termination at the end of round $\numRounds$:] \quad
\label{stp:normalTerminationProcess}   
\begin{enumerate}
\item The simulator gets $\outValue$ from the trusted party $T_{\MPCWithDealer}$.
\item $\S$ constructs all the singed shares of the inner secret-sharing scheme 
for each $J \subseteq \range{\partNum}\setminus \abortedAfterPre$ s.t. $\goodNum \leq \size{J} \leq \badNum$ and for each honest party 
$p_j\in \superSet{J}$ as follows.


For each $J \notin \J$, the simulator
$S$ selects uniformly at random $\size{J\setminus \badI}$ shares of $\outValue$ uniformly at 
                                random over all possible selections of $\size{J\setminus \badI}$ shares that together with
                                the $\size{J\cap \badI}$ given shares 
                                $\set{\valueInnerSecretSharingShareSigned{i}{J}{j}: j \in \badI}$
                                (produced in \stepref{innerSharesConstruction} in the simulation of the preliminary phase)
                                are a sharing of $\outValue$ in a $\size{J}$-out-of-$\size{J}$ secret sharing scheme.\\
                                (This is possible according to the property of Shamir's scheme)\\
Denote these shares by 
$\set{\valueInnerSecretSharingShareFromT{\numRounds}{J}{j}}$. 

For each share $\valueInnerSecretSharingShareFromT{\numRounds}{J}{j}$, the simulator 
concatenates the corresponding identifying details, and signs 
them to obtain: 
$\valueInnerSecretSharingShareSignedFromT{\numRounds}{J}{j}\leftarrow
   (\valueInnerSecretSharingShareFromT{\numRounds}{J}{j}, \numRounds, J, j ,
   \operatorname{Sign}((\valueInnerSecretSharingShareFromT{\numRounds}{J}{j}, 
   \numRounds, J, j) ,K_{\rm sign})).$

\item \label{stp:sendMSGinTerminationProcess} For each 
honest party $p_j$, the simulator $\S$ sends to $\A$ the shares
$\valueInnerSecretSharingShareSignedFromT{\numRounds}{J}{j}$ for all subsets $J$, such 
that $p_j\in\superSet{J}$.
\item The simulator $\S$ outputs the sequence of messages exchanged between 
        $\S$ and the adversary $\A$ and halts.
\end{enumerate}
\end{description}

\subsection{Proving the Correctness of Protocol $\MPC_\numRounds$ and Protocol $\MPCRange_\numRounds$}

It can be proved that Protocol $\MPC_\numRounds$ is a secure implementation of the
(ideal) functionality of the dealer's in Protocol $\MPCWithDealer_\numRounds$.
That is, 
\begin{lemma}
\label{lem:secureImplementation}
Let $\badNum < 2\partNum/3$.
If enhanced trap-door permutations exist,
then Protocol $\MPC_{\numRounds}$ presented in
\secref{eliminatingDealer}, is a computationally-secure implementation (with full security)
of the dealer functionality in Protocol  
$\MPCWithDealer_{\numRounds}$.
\end{lemma}

In ~\cite{BOO10FULL}, a similar framework to the one used in this paper is used:
first a protocol with a dealer for the coin-tossing problem is presented 
and, then, a real-world protocol that is a computationally-secure implementation (with full security)
of the dealer functionality is described.
In~\cite{BOO10FULL}, a simulator for this protocol is given.
This simulator is similar to the simulator described in \appref{OfflineProof},
than a full proof for the simulator is provided.
As the proof is very similar to the proof of our simulator, we omit the proof.

To conclude the proof, as $\MPCWithDealer_\numRounds$ is a $1/p$-secure implementation of $\F$ and
 $\MPC_\numRounds$ is a secure implementation of the
(ideal) functionality of the dealer in Protocol $\MPCWithDealer_\numRounds$,
by the composition theorem of Canetti~\cite{Can00} we conclude that 
$\MPC_\numRounds$ $1/p$-secure implementation of $\F$.
That is, \thmref{mainDomain} is proved.

Next, we claim that $\MPCRange_\numRounds$ is a secure implementation of the
(ideal) functionality of the dealer in Protocol $\MPCWithDealerRange_\numRounds$.
That is,
\begin{lemma}
\label{lem:MPCRange}
Let $\badNum < 2\partNum/3$.
If enhanced trap-door permutations exist,
then Protocol $\MPCRange_{\numRounds}$ described in
\secref{polynomialRange}, is a computationally-secure implementation (with full security)
of the dealer functionality in Protocol  
$\MPCWithDealerRange_{\numRounds}$.
\end{lemma}

\begin{proof}
Recall that the only difference between Protocol $\MPC_\numRounds$ and Protocol $\MPCRange_\numRounds$ 
is in the way that the values that the parties see prior round $i^\star$ are produced, i.e.,
the difference is in Functionality $\MultiShareGenDomain_\numRounds$.
Specifically, in \secref{polynomialRange} we presented a modification in \stepref{beforeIStar} in Functionality $\MultiShareGenDomain_\numRounds$ in order to get Protocol $\MPC_\numRounds$ from Protocol $\MPCRange$.
Now, observe that the simulator presented above does not refer to \stepref{beforeIStar} of Functionality $\MultiShareGenDomain_\numRounds$ in any step.
Therefore, the simulator presented in \appref{OfflineProof} for Protocol $\MPC_\numRounds$
is also a simulator for Protocol $\MPCRange_\numRounds$.
\end{proof}

\medskip
\clmref{equalValuesFromRange}  and \lemref{MPCRange} imply \thmref{mainRange}. 

\end{document}